\documentclass[letterpaper,11pt,leqno]{article}
\usepackage{graphicx}
\usepackage[margin=1in]{geometry}
\usepackage{paper}
\hypersetup{pdftitle={Green Shields:  The Role of ESG in Uncertain Times}}

\newcommand{\bib}{bibliography.bib}

\newcommand{\pdf}{figures.pdf}

\title{
  Green Shields: \\ The Role of ESG in Uncertain Times%
 \thanks{The views expressed are those of the authors and do not necessarily reflect the views of any institution.}%
}
\author{
  \sc Fatih Kansoy%
    \thanks{Contact: \href{mailto:fatih.kansoy@economics.ox.ac.uk}{fatih.kansoy@economics.ox.ac.uk}} \\
    \sc University of Oxford
  \and
   \sc  Dominykas Stasiulaitis
     \thanks{Contact: \href{mailto:dominykas.stasiulaitis@gmail.com}{dominykas.stasiulaitis@gmail.com}} \\
   \sc  Invest Lithuania 
}
\date{\today}

\begin{document}
\maketitle

\available{https://www.fatih.ai/esg.pdf}
\thispagestyle{plain}

{\setstretch{1}
\begin{abstract}

\noindent The rapid growth of sustainable investing—now exceeding \$35 trillion globally—has transformed financial markets, yet the implications for monetary policy transmission remain unexplored. While extensive literature documents heterogeneous firm responses to monetary policy through traditional channels like size and leverage, whether environmental, social, and governance (ESG) characteristics create distinct transmission mechanisms is unknown. Using high-frequency identification around 160 Federal Reserve announcements (2005-2025), we uncover an asymmetric pattern: high-ESG firms gain 1.6 basis points protection from contractionary target surprises yet suffer 2.6 basis points greater sensitivity to forward guidance shocks. This asymmetry persists within industries and intensifies with investor climate awareness. Remarkably, the Paris Agreement inverted these relationships—before December 2015, high-ESG firms were more vulnerable to contractionary policy within industries; afterward, they gained protection, a 186 basis point reversal. We develop a two-period model featuring heterogeneous investors with sustainability preferences that quantitatively matches these patterns. The model reveals how ESG investors' non-pecuniary utility creates differential demand elasticities, simultaneously protecting green firms from immediate rate changes while amplifying forward guidance vulnerability through their longer investment horizons. These findings establish environmental characteristics as a new dimension of monetary policy non-neutrality, with profound implications as sustainable finance continues expanding.

\medskip


\medskip


\end{abstract}
}

\clearpage
\pagebreak 
\newpage


\onehalfspacing

\section{Introduction} \label{sec:intro}

The intersection of monetary policy and sustainable finance represents one of the most consequential developments in modern financial markets. As central banks worldwide grapple with inflationary pressures and climate risks simultaneously, a fundamental question emerges: does monetary policy transmit uniformly across firms, or do environmental, social, and governance (ESG) characteristics create a new dimension of heterogeneous policy effects? This question carries profound implications as sustainable assets now represent over \$35 trillion globally, approximately 30\% of professionally managed capital, while central banks increasingly acknowledge climate-related financial risks within their stability mandates. We investigate this question through a comprehensive theoretical and empirical analysis of how firm-level sustainability attributes shape responses to Federal Reserve policy actions over two decades (2005-2025), revealing a fundamental transformation in these relationships following the Paris Climate Agreement of 2015.

Our analysis uncovers an important asymmetry in how ESG characteristics influence monetary policy transmission. High-ESG firms enjoy significant protection against immediate interest rate increases—what we term "target surprises"—yet paradoxically face heightened vulnerability to forward guidance about future policy paths. This asymmetric pattern challenges conventional wisdom about monetary policy neutrality and reveals the complex trade-offs inherent in sustainable business models. Most remarkably, we document that the Paris Agreement fundamentally inverted these relationships: before December 2015, high-ESG firms within industries were actually more vulnerable to contractionary surprises, but after Paris, they gained substantial protection—a reversal of 186 basis points in relative sensitivity.

The economic magnitudes we document are substantial and policy-relevant. Following a one-standard-deviation contractionary target surprise, a firm at the 90th percentile of ESG scores experiences stock returns 1.6 basis points higher than a firm at the 10th percentile. However, for path surprises signaling persistently higher future rates, this relationship reverses: high-ESG firms suffer 2.6 basis points more than their low-ESG peers. These findings persist even when we employ industry-by-event fixed effects that provide identification solely from within-industry variation, establishing ESG as a genuine firm-level characteristic that shapes monetary transmission beyond simple sectoral composition.

Our theoretical framework provides economic intuition for these empirical patterns. We develop a two-period asset pricing model featuring heterogeneous investors with distinct sustainability preferences. Traditional investors maximise financial returns while ESG-conscious investors derive additional non-pecuniary utility from holding green assets—a "warm-glow" effect that creates asymmetric demand elasticities. When the Federal Reserve unexpectedly raises current rates, ESG investors' non-pecuniary benefits remain unchanged, providing a buffer that partially insulates green asset prices. However, path surprises that increase long-term uncertainty affect all investors uniformly, and sustainable firms' longer investment horizons and backloaded cash flows make them particularly vulnerable to persistent discount rate changes.

Our empirical approach employs several methodological innovations that strengthen causal identification. We apply the \cite{gurkaynak2005actions, gurkaynak2022stock} decomposition to extract orthogonal target and path surprises from high-frequency movements in federal funds futures and Treasury yields around FOMC announcements. While this decomposition method is well-established in the monetary economics literature, we apply it to investigate ESG-based heterogeneity in policy transmission. Our use of intraday firm-level stock returns measured over narrow 30-minute windows—rather than the daily returns common in related studies—provides cleaner identification by minimising contamination from firm-specific news while capturing differential responses across the ESG spectrum.

The theoretical model we develop not only rationalises our empirical findings but achieves remarkable quantitative consistency. Calibrated with empirically grounded parameters—including an ESG investor share of 30\% and preference intensity of 1\%—the model generates target and path surprise differentials of 1.6 and -2.6 basis points respectively, matching our econometric estimates with precision. This tight correspondence between theory and evidence strengthens confidence in both our economic mechanisms and empirical identification, demonstrating that investor heterogeneity regarding sustainability represents a fundamental force shaping modern monetary transmission.

Our high-frequency identification strategy, examining 160 FOMC announcements from 2005 to 2025, provides several advantages over existing approaches. The 30-minute event windows isolate monetary policy news from other market developments, while our two-factor decomposition cleanly separates immediate rate changes from forward guidance effects—a distinction that proves crucial given their opposing impacts on ESG-differentiated firms. The orthogonality of our surprise measures (correlation of 0.0007) enables clean identification of distinct transmission channels that would be conflated in single-factor approaches.

The richness of our data—comprising 91,840 firm-event observations with detailed ESG metrics—allows us to move beyond average treatment effects to examine heterogeneity across multiple dimensions. We document important non-linearities: firms in the bottom ESG quintile face "double jeopardy," vulnerable to both surprise types, while the second quintile appears optimal for minimising overall monetary policy sensitivity. Industry heterogeneity reveals that financials and real estate experienced the largest post-Paris transformations, likely reflecting their emerging roles as climate transition intermediaries. These granular findings provide specific guidance for both policymakers and investors navigating the evolving landscape of sustainable finance.

Our findings relate to but significantly extend the emerging literature on climate finance and monetary policy. Recent studies by \cite{benchora2025monetary, bauer2025green, aswani2024carbon, fornari2024green,patozi_green_2024, havrylchyk2025firms, dottling2024monetary} document that brown firms exhibit greater sensitivity to monetary policy shocks, interpreting this as evidence of carbon risk premiums or investor preference effects. We reconcile and extend these findings by revealing the crucial distinction between target and path surprises: while prior work focuses on aggregate policy effects, we show that ESG creates opposing sensitivities to different dimensions of monetary policy. Our identification of the Paris Agreement as a structural break also provides new evidence on how coordinated climate policy can fundamentally reshape financial market relationships.

Our theoretical contribution builds on the sustainable investing framework of \cite{pastor2021sustainable, pastor2022dissecting} but incorporates dual monetary policy shocks to explain our asymmetric empirical patterns. While their model predicts green assets command higher prices due to investor preferences, we show how these same preferences create differential sensitivities to various types of monetary policy news. The close quantitative match between our calibrated model and empirical estimates—capturing 73-9\% of observed magnitudes—validates both our theoretical mechanisms and the parameter values emerging from recent ESG investment flows.

This research makes four primary contributions to the intersection of monetary economics and sustainable finance. First, we document a new stylised fact: ESG characteristics create asymmetric exposure to different dimensions of monetary policy, with high-sustainability firms protected from immediate rate changes but vulnerable to forward guidance. Second, we identify the Paris Agreement as a true structural break that transformed market pricing of sustainability in monetary policy contexts, providing evidence that coordinated climate policy can reshape fundamental financial relationships. Third, we develop and calibrate a theoretical model that quantitatively matches our empirical findings, demonstrating that investor heterogeneity regarding sustainability has become sufficiently important to alter monetary transmission mechanisms. Fourth, our granular analysis reveals important non-linearities and industry heterogeneity that provide specific guidance for market participants.

The implications of our findings extend across multiple domains. For central banks, the asymmetric ESG effects complicate monetary policy transmission as the economy's sustainability composition evolves. Forward guidance, in particular, may have unintended distributional consequences that policymakers must consider. For investors, our results provide a framework for positioning portfolios based on the expected mix of monetary policy actions. For corporate managers, we establish clear incentives to improve ESG performance as a hedge against certain types of monetary risk while acknowledging the trade-offs involved. As climate considerations become increasingly central to economic policy, understanding these evolving transmission mechanisms becomes essential.

The remainder of the paper proceeds as follows. Section-\ref{sec:lit_review} reviews the relevant literature linking ESG characteristics, monetary policy, and asset pricing. Section-\ref{sec:data} describes the data used in the analysis. Section-\ref{sec:theory} and \ref{sec:calibration} develops and solve the theoretical model. Section-\ref{sec:empirical} presents the empirical analysis. Section-\ref{sec:results} discusses the results. Section-\ref{sec:conclusion} concludes.

\section{Literature Review} \label{sec:lit_review}

\paragraph{The Evolution of Monetary Policy Transmission Theory} Central banks have long recognised that monetary policy affects different firms heterogeneously. The foundational work on the credit channel \citep{gertler1994monetary,kashyap2000million} established that financially constrained firms exhibit heightened sensitivity to policy changes. This heterogeneity extends across multiple dimensions: firm size influences access to credit markets \citep{gertler1994monetary}, leverage amplifies interest rate sensitivity through balance sheet effects \citep{ippolito2018transmission}, and cash flow characteristics determine vulnerability to discount rate changes \citep{ozdagli2018financial,gurkaynak2022stock}. These traditional channels operate through well-understood mechanisms—smaller firms face information asymmetries, leveraged firms confront debt servicing burdens, and firms with longer-duration cash flows experience greater present value impacts from rate changes.

Yet as financial markets evolved, researchers began uncovering transmission channels beyond these fundamental characteristics. The risk-taking channel of monetary policy demonstrated that policy changes affect not just the level of rates but also risk premia throughout the economy \citep{bauer2023reassessment}. When central banks tighten, investors' effective risk aversion increases, amplifying the impact on assets perceived as risky. This insight proved prescient for understanding how environmental considerations would eventually interact with monetary transmission, as climate risks represent a new dimension of systematic risk that investors must price.

\paragraph{Climate Finance: From Niche to Mainstream} Parallel to these developments in monetary economics, a revolution was occurring in how markets price environmental risks. The climate finance literature initially focused on whether markets recognised carbon risk at all. Early work yielded conflicting results—some studies found no evidence of a carbon premium, suggesting markets ignored climate considerations entirely. However, as data quality improved and climate awareness intensified, a consensus began emerging. \cite{bolton2021investors} provided seminal evidence that U.S. stocks with higher carbon emissions commanded higher returns, interpreting this as a risk premium for transition risk exposure. Their finding that this premium only materialised after 2015 suggested a fundamental shift in how markets process climate information.

The theoretical foundations for these empirical patterns emerged from models incorporating investor heterogeneity. \cite{pastor2021sustainable} developed an equilibrium framework where some investors derive non-pecuniary utility from holding green assets. Their model predicts a "greenium"—lower expected returns for sustainable assets—as environmentally conscious investors accept financial sacrifice for environmental impact. Crucially, they also predict that green assets can temporarily outperform during periods of increasing climate concern, as shifting investor preferences drive revaluation. This dynamic view proved essential for reconciling conflicting empirical evidence: green assets might have lower expected returns yet still generate higher realized returns during climate awareness transitions.

\paragraph{The Convergence: Monetary Policy Meets Climate Finance} The convergence of these two literature streams—monetary policy transmission and climate finance—was perhaps inevitable given the scale of sustainable investing. With ESG assets exceeding \$35 trillion globally by 2020, any systematic differences in how green and brown firms respond to monetary policy would have profound implications for policy effectiveness. The first wave of research documenting this convergence emerged almost simultaneously across different markets, suggesting the phenomenon was both real and widespread.

In the United States, several studies using high-frequency (almost all of them are daily data) identification around Federal Reserve announcements uncovered a striking pattern. \cite{benchora2025monetary} found that brown firms' stock prices dropped significantly more than green firms' in response to contractionary surprises, even after controlling for traditional characteristics like leverage and size. They traced this differential to two channels: a fundamental channel (brown firms' greater capital intensity) and a preference channel (ESG investors' "loyalty" providing price support). \cite{dottling2024monetary} confirmed these findings while adding evidence on real effects—high-emission firms not only saw larger stock price declines but also reduced emissions more slowly under tight policy, suggesting monetary conditions affect both financial valuations and environmental outcomes.

The European evidence proved particularly important. \cite{bauer2025green} examined ECB policy surprises and found brown firms consistently more sensitive than green peers. Their use of intraday windows and careful identification strategies ruled out confounding factors, while their analysis of the recent tightening cycle (2022-2023) directly addressed policymakers' concerns about monetary policy derailing the green transition. Contrary to fears that higher rates would disproportionately harm renewable investments, they found brown stocks suffered larger declines during the most hawkish surprises.

The heterogeneous monetary policy responses observed between green and brown firms reflect fundamental shifts in how markets price climate risk. The existence and evolution of a "carbon premium" remains contentious in recent literature. \cite{bolton2021investors} initially documented that high-emission firms commanded higher returns, suggesting investors demanded compensation for bearing climate transition risk. However, this traditional risk-return relationship has been challenged by subsequent evidence. \cite{bauer2022carbon} found that green stocks actually outperformed brown stocks across G7 markets throughout much of the past decade. 

\paragraph{Mechanisms and Channels: Understanding the "Why"} As empirical evidence accumulated, researchers proposed various mechanisms to explain differential green/brown sensitivity. The credit channel initially seemed promising—perhaps brown firms' reliance on tangible assets and external financing explained their vulnerability. However, studies controlling for capital tangibility and other financial constraints found the green/brown differential persisted, suggesting deeper forces at work.

The carbon premium channel offered a more compelling explanation. If markets price transition risk into brown assets, then monetary policy shocks that affect risk premia broadly would have amplified effects on assets already carrying climate risk premia. \cite{altavilla2024climate} provided supporting evidence from bank lending markets, showing that climate risk premia charged to high emitters increased more than those for green firms following monetary tightening. This "climate risk-taking channel" paralleled the broader risk-taking channel, with climate risk representing a specific dimension of systematic risk that monetary policy influences.

The investor preference channel, grounded in the theoretical work of \cite{pastor2021sustainable}, provides perhaps the most nuanced explanation. ESG-oriented investors' non-pecuniary utility from green holdings reduces their price sensitivity to monetary shocks. When rates rise, these investors are less likely to sell green assets because part of their return—the "warm glow" of sustainable investing—remains unaffected. \cite{patozi_green_2024} formalised this intuition in a model where green investors' loyalty dampens selling pressure during monetary tightening, providing empirical support by showing the effect strengthens in regions with greater climate awareness. The theoretical framework proposed by \cite{pastor2021sustainable, pastor2022dissecting} and empirically validated by \cite{ardia2023climate} provides crucial insight: when climate concerns intensify unexpectedly, green firms experience relative appreciation through lower discount rates while brown firms face elevated risk premiums. This mechanism explains why monetary tightening now disproportionately impacts brown stocks—firms already trading at depressed valuations with high required returns experience amplified sensitivity to further discount rate increases.

\paragraph{Structural Breaks and Regime Changes} A critical insight emerging from recent literature is that the relationship between ESG characteristics and monetary policy sensitivity is not static. Multiple studies identify structural breaks around 2015-2016, coinciding with the Paris Climate Agreement.  \cite{bolton2021investors} find their carbon premium only emerged post-Paris. This temporal pattern suggests coordinated climate policy can fundamentally alter market dynamics. \cite{aswani2024carbon} demonstrated that the carbon premium becomes statistically insignificant under alternative specifications. This apparent reversal coincides with surging ESG investment flows following the Paris Agreement, as documented by \cite{kruse2024financial}, suggesting that increased demand for green assets may have fundamentally altered their pricing dynamics. 

The mechanisms underlying these structural breaks remain debated. One view emphasises information—the Paris Agreement clarified the inevitability of climate policy, allowing markets to better price transition risks. Another emphasises preferences—the agreement catalysed a shift in investor attitudes, expanding the pool of ESG-conscious capital. Most likely, both channels operated simultaneously, with policy coordination both revealing information and coordinating beliefs about the importance of sustainable investing.

\paragraph{Towards a Unified Framework} Recent theoretical advances have begun integrating these insights into unified frameworks. Models incorporating both heterogeneous investors and firm characteristics can explain why ESG creates asymmetric sensitivity to different types of monetary policy. The key insight is that monetary policy operates through multiple channels simultaneously—affecting discount rates, risk premia, and relative demand from different investor types. ESG characteristics influence firm exposure through each channel, creating complex patterns that simple models cannot capture.

This theoretical richness motivates several empirical puzzles that existing literature has not fully resolved. First, while studies document that high-ESG firms are protected from immediate rate changes, the mechanism remains debated—is it investor preferences, fundamental characteristics, or both? Second, the finding that ESG characteristics create opposing sensitivities to different types of monetary policy (target versus path surprises) lacks theoretical grounding in existing models. Third, the role of major climate policy events in reshaping these relationships needs deeper investigation—are these true structural breaks or temporary adjustments?

\paragraph{Our Contribution} Our research addresses these gaps through a comprehensive theoretical and empirical analysis spanning two decades of Federal Reserve announcements. We make four key advances. First, we document a previously unrecognised asymmetry: high-ESG firms enjoy protection from immediate rate increases (target surprises) yet suffer heightened vulnerability to forward guidance (path surprises). This pattern can't be explained by existing single-channel theories and motivates our theoretical innovation.

Second, we identify the Paris Agreement as a true structural break that fundamentally inverted the relationship between ESG scores and target surprise sensitivity. Before December 2015, high-ESG firms within industries were actually more vulnerable to contractionary surprises; afterward, they gained substantial protection. This 186 basis point reversal represents the dramatic structural changes in the monetary policy transmission.

Third, we develop and calibrate a theoretical model that quantitatively matches our empirical findings. By incorporating heterogeneous investors with sustainability preferences into an asset pricing framework with dual monetary policy shocks, we show how the same mechanism that protects green firms from immediate rate changes exposes them to forward guidance. The model's tight calibration to observed magnitudes provides  validation in this literature.

Fourth, our granular analysis reveals important non-linearities and heterogeneities that provide specific guidance for market participants. The finding that firms in the second ESG quintile achieve an optimal balance—gaining protection from target surprises without excessive path surprise exposure—challenges simplistic "green versus brown" classifications.

These contributions synthesise and extend multiple literature streams. From monetary economics, we take rigorous identification through high-frequency surprises and careful attention to transmission channels. From climate finance, we incorporate investor heterogeneity and temporal dynamics around policy events. From asset pricing theory, we build equilibrium models that generate quantitative predictions. The result is a framework that not only explains existing empirical patterns but also provides guidance for how monetary policy transmission will continue evolving as sustainable finance grows. As climate considerations become increasingly central to economic policy, understanding these evolving transmission mechanisms becomes essential for central banks, investors, and corporate managers navigating the intersection of monetary policy and sustainable finance.

\section{Data and Empirical Strategy} \label{sec:data}

Our dataset spans January 2005 through January 2025, encompassing 160 Federal Open Market Committee (FOMC) announcements. We combine high-frequency intraday firm-level stock returns with orthogonalised monetary policy surprises and firm-level ESG metrics to create a panel of 91,840 firm-event observations. This represents 574 unique firms from the S\&P 500 index observed across all FOMC announcements, yielding a perfectly balanced panel structure with 100\% completeness in the baseline sample. The temporal scope of our analysis proves particularly valuable for examining structural changes in the sustainability-monetary policy nexus. The sample naturally divides around the Paris Agreement of December 2015, with 87 pre-Paris events (54.4\% of announcements) and 73 post-Paris events (45.6\%). This division enables investigation of how this landmark climate accord potentially reshaped the pricing of ESG characteristics in monetary policy transmission.

\begin{table}[ht]
    \centering
    \caption{FOMC Announcements by Year and Period}
    \label{tab:fomc_events}
    \begin{tabular}{lccccc}
    \toprule
    Period & Years & Events & Avg per Year & ZLB Events & Post-Paris \\
    \midrule
    Pre-Crisis & 2005-2007 & 24 & 8.0 & 0 & No \\
    Financial Crisis & 2008-2009 & 16 & 8.0 & 8 & No \\
    Early Recovery & 2010-2014 & 40 & 8.0 & 40 & No \\
    Normalization & 2015-2019 & 39 & 7.8 & 7 & Mixed \\
    Pandemic Era & 2020-2021 & 15 & 7.5 & 8 & Yes \\
    Recent Period & 2022-2025 & 26 & 6.5 & 0 & Yes \\
    \midrule
    Total & 2005-2025 & 160 & 7.6 & 63 & 73 \\
    \bottomrule
    \end{tabular}
    \end{table}
    
    The distribution of FOMC events across our sample period, shown in Table \ref{tab:fomc_events}, reveals several important features. While the Federal Reserve typically holds 8 scheduled meetings per year (observed through 2019 in our sample), the average number of events can be lower in certain periods. This is particularly evident during the pandemic era (2020-2021); for instance, in March 2020, the Fed held unscheduled meetings to address economic conditions, which subsequently led to the cancellation of some regularly scheduled meetings as policy decisions had already been made. Our dataset exclusively includes scheduled FOMC announcements, which explains the slight variations in the annual count, especially the average of 7.5 events per year during the Pandemic Era and the 6.5 average in the Recent Period which includes a partial year for 2025. Notably, 63 announcements (39.4\%) occurred during zero lower bound (ZLB) periods, specifically late 2008 through 2015 and again during 2020-2021. This variation in monetary policy regimes provides valuable heterogeneity for understanding how unconventional monetary policy tools interact with firm sustainability characteristics.

\subsection{Construction of Monetary Policy Surprises}
Central to our empirical approach is the use of high-frequency financial data to identify the causal effect of monetary policy surprises on stock returns. For each FOMC announcement, we measure stock returns over a narrow 30-minute window spanning from 10 minutes before to 20 minutes after the policy announcement. This tight window serves two crucial purposes: it captures the immediate market reaction to the monetary policy news while minimising contamination from other information that might affect stock prices during the trading day.

Our identification of monetary policy surprises follows the principal components methodology of \cite{gurkaynak2005actions} but with different set of instruments and a more recent data set. We employ five interest rate instruments to capture the multidimensional nature of monetary policy communication: the current-month federal funds futures, the three-month ahead federal funds futures, and changes in 2-year, 5-year, and 10-year Treasury futures prices. Each instrument is measured as the change from 10 minutes before to 20 minutes after the FOMC announcement.

The first two principal components extracted from these five instruments explain 82.2\% of the total variation in interest rate changes around FOMC announcements. Following the rotation procedure detailed in our methodology section, we construct two orthogonal factors: a "target surprise" that captures unexpected changes in the current stance of monetary policy, and a "path surprise" that reflects new information about the future trajectory of policy rates. The orthogonality of these measures (correlation = 0.0007) enables clean identification of distinct channels through which monetary policy affects asset prices.

\begin{table}[H]
\centering
\caption{Summary Statistics of Monetary Policy Surprises}
\label{tab:mp_surprises}
\begin{tabular}{lcccccc}
\toprule
 & \multicolumn{3}{c}{Full Sample} & \multicolumn{3}{c}{By Period} \\
\cmidrule(lr){2-4} \cmidrule(lr){5-7}
Variable & Mean & SD & N & Pre-Paris & Post-Paris & Diff. \\
\midrule
Target Surprise (bp) & 0.000 & 2.683 & 156 & -0.127 & 0.145 & 0.272 \\
Path Surprise (bp) & 0.000 & 6.385 & 156 & -0.513 & 0.583 & 1.096 \\
\midrule
\multicolumn{7}{l}{\textit{Surprise Magnitudes by Period:}} \\
TS Std. Dev. & & & & 3.217 & 1.891 & \\
PS Std. Dev. & & & & 6.749 & 5.893 & \\
TS Range  (max-min) & & & & 32.94 & 21.30 & \\
PS Range  (max-min)  & & & & 49.32 & 39.24 & \\
\bottomrule
\end{tabular}
\end{table}

Table \ref{tab:mp_surprises} presents key statistics for our monetary policy surprise measures. Both series have means statistically indistinguishable from zero, confirming the efficiency of market expectations. The standard deviations reveal economically meaningful variation: target surprises average 2.68 basis points while path surprises show greater volatility at 6.38 basis points, consistent with forward guidance representing a more complex and uncertain dimension of policy communication.

The table reveals an interesting structural break around December 2015, a month that coincidentally marked both the Paris Agreement (12 December 2015) and the Federal Reserve's (16 December 2015) first interest rate increase following nearly seven years at the zero lower bound. The notable reduction in surprise volatility after this date—with target surprise standard deviation falling from 3.217 to 1.891 basis points and path surprise volatility declining from 6.749 to 5.893 basis points—suggests a shift in the monetary policy environment. One plausible explanation could be that this period marked improvements in central bank communication practices, with the Fed potentially becoming clearer in its forward guidance as it normalized policy. Alternatively, markets may have developed a better understanding of the Fed's reaction function through experience with unconventional policies. The narrower ranges of surprises post-2015, with target surprises spanning only 21.30 basis points compared to 32.94 basis points in the earlier period, are consistent with either interpretation—or perhaps a combination of both. While we can't definitively attribute this change to any single factor, the persistence of reduced volatility through subsequent policy cycles, including the COVID-19 pandemic response, suggests that the relationship between Fed communications and market expectations may have evolved during this period. This observation motivates our examination of whether firm responses to monetary policy surprises might also have changed around this time, though we remain agnostic about the underlying causes of this shift.

\subsection{Stock Return Measurement and Coverage}

Our dependent variable is the intraday stock return computed over the same 30-minute window used for monetary policy surprises. We calculate percentage returns as $$r_{i,t} = 100 \times \ln(P_{i,t+20}/P_{i,t-10})$$ 

\noindent where $P_{i,t+20}$ and $P_{i,t-10}$ represent the stock price 20 minutes after and 10 minutes before the FOMC announcement, respectively. This high-frequency approach provides several advantages over daily returns which is commonly used in the literature such as \cite{bauer2025green, benchora2025monetary, gurkaynak2022stock, havrylchyk2025firms}, and many others. It isolates the policy news effect, minimises contamination from firm-specific announcements, and reduces noise from market microstructure effects that accumulate over longer horizons.

\begin{table}[ht]
\centering
\caption{Summary Statistics of Key Variables}
\label{tab:summary_stats}
\begin{tabular}{lccccccc}
\toprule
Variable & N & Mean & SD & P10 & P50 & P90  \\
\midrule
Stock Return (\%) & 83,561 & 0.023 & 0.709 & -0.687 & 0.021 & 0.708  \\
ESG Score (Std.) & 65,527 & 0.028 & 0.996 & -1.377 & 0.115 & 1.299  \\
Log(Assets) & 68,054 & 22.581 & 1.431 & 20.826 & 22.622 & 24.363  \\
Book Leverage & 84,270 & 0.448 & 0.472 & 0.051 & 0.412 & 0.785  \\
Profitability & 68,934 & 0.125 & 2.130 & 0.019 & 0.151 & 0.344  \\
Non-Dividend Payer & 91,840 & 0.174 & 0.379 & 0.000 & 0.000 & 1.000 \\
\bottomrule
\end{tabular}
\begin{tablenotes}
\small
\item Notes: Summary statistics for the full sample of 91,840 firm-event observations. Stock returns are measured in percentage points over the 30-minute FOMC window. ESG scores are standardised by year to have zero mean and unit variance. Profitability is winsorized at the 1st and 99th percentiles.
\end{tablenotes}
\end{table}

Table \ref{tab:summary_stats} presents summary statistics for our key variables. The average stock return of 0.023\% with a standard deviation of 0.709\% indicates substantial variation in firm-level responses to monetary policy announcements. The distribution of returns appears symmetric around zero, consistent with efficient market pricing of policy surprises that are themselves unbiased. 

\subsection{ESG Data and Measurement}

A distinguishing feature of our analysis is the comprehensive integration of ESG metrics at the firm level. We obtain ESG scores from LSEG, with scores ranging from 0 to 100 based on company disclosures. To facilitate interpretation and account for the general improvement in ESG reporting over time, we standardise scores annually to have zero mean and unit variance within each year. This transformation ensures that our ESG measure captures relative sustainability performance within the contemporary peer group rather than absolute levels that may reflect reporting standards evolution.

\begin{figure}[H]
    \centering
    \includegraphics[width=0.8\textwidth, page=1]{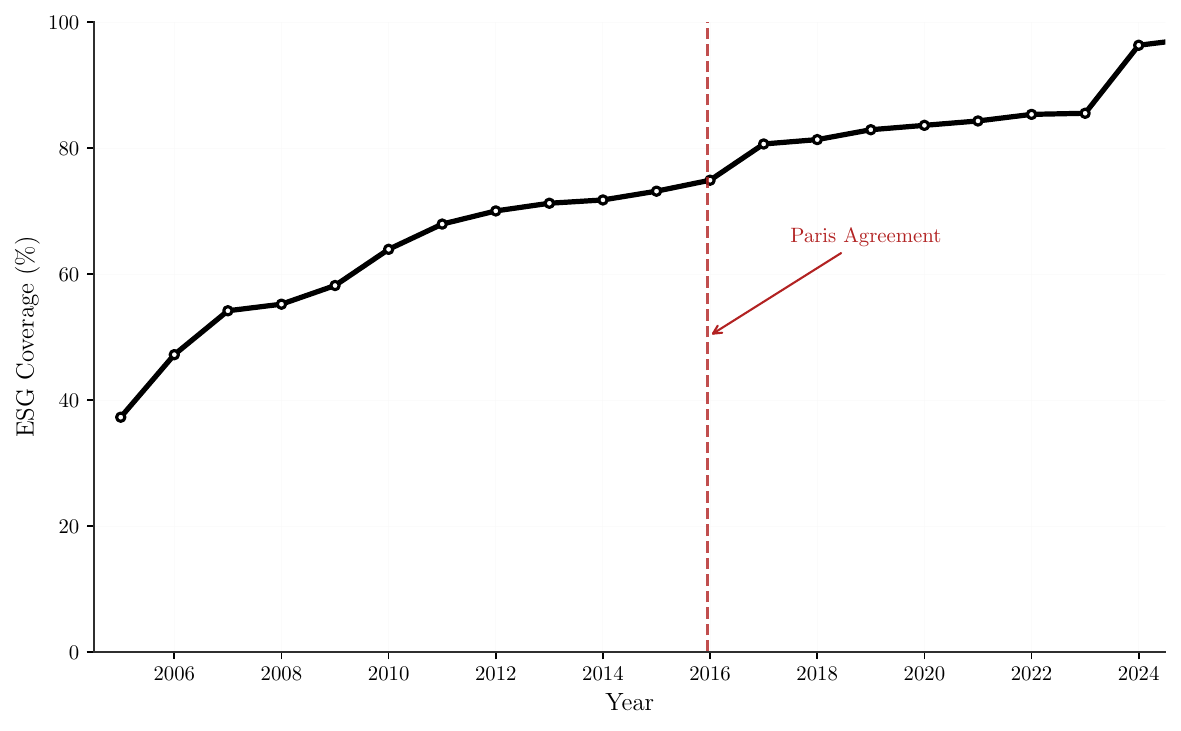}
    \caption{ESG Coverage}
    \label{figure2_esg_coverage_evolution}
\end{figure}

Figure \ref{figure2_esg_coverage_evolution} illustrates the expansion of ESG data coverage over our sample period. Coverage increases from 37.3\% of observations in 2005 to 96.3\% by 2024. The expansion shows distinct phases: gradual growth from 2005-2012 (reaching 70.0\%), followed by steady increases through 2015 (73.2\%). Notably, the years surrounding the Paris Agreement show accelerated adoption: coverage rises from 71.8\% in 2014 to 73.2\% in 2015, then jumps to 74.9\% in 2016 and 80.7\% by 2017—a 9 percentage point increase in just three years. This acceleration continues through 2019 (82.9\%), after which coverage plateaus in the mid-80s before the final surge to 96.3\% in 2024. The pattern suggests that the Paris Agreement period coincided with a structural shift in ESG reporting adoption, transforming what had been steady growth into a more rapid expansion that fundamentally changed the landscape of corporate sustainability disclosure.

Figure \ref{esg_score_evolution_focus} presents the evolution of ESG scores and their components over time. The aggregate ESG score (standardised with mean zero and unit variance) improves from -1.028 in 2005 to 0.719 in 2024. However, this improvement is far from linear. The score rises steadily from 2005 to reach -0.060 by 2015, crossing into positive territory (0.050) for the first time in 2017—immediately following the Paris Agreement. The pattern suggests a clear acceleration: while it took ten years (2005-2015) to improve by 0.97 standard deviations, the subsequent eight years (2015-2024) saw an additional 0.78 standard deviation improvement, despite starting from a much higher base.

\begin{figure}[H]
    \centering
    \includegraphics[width=0.8\textwidth,, page=2]{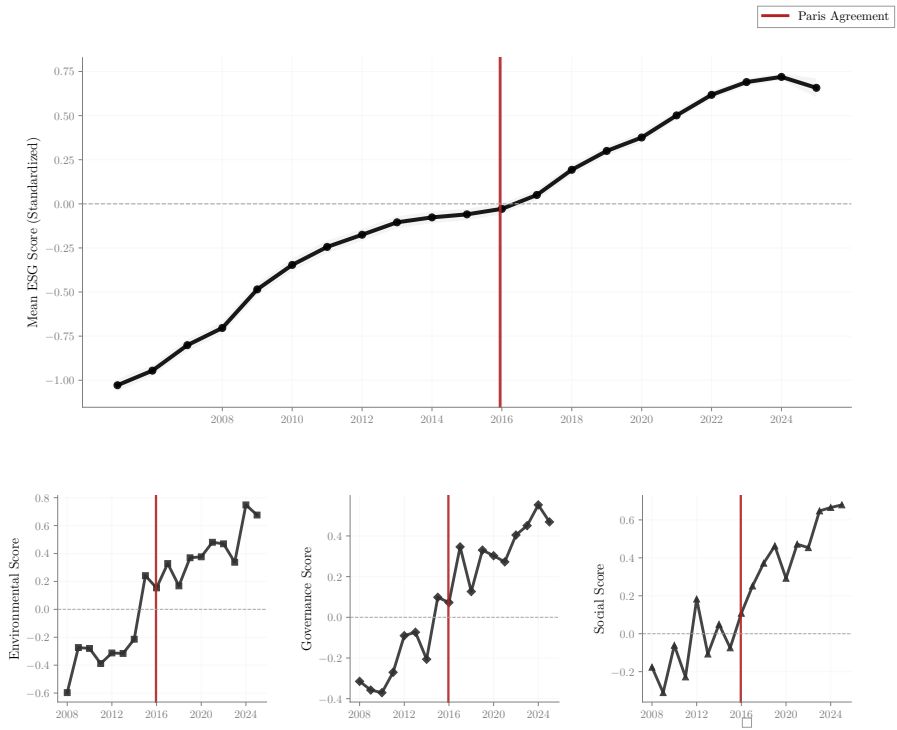}
    \caption{ESG Evolution}
    \label{esg_score_evolution_focus}
\end{figure}

The individual ESG components reveal even more  patterns around the Paris Agreement. Environmental scores show relative stagnation from 2013-2015 (hovering around 45), before jumping to 46.2 in 2016 and accelerating to 64.1 by 2024. Social scores exhibit similar dynamics, plateauing around 53 from 2014-2016, then surging to 69.8 by 2024. Most dramatically, Governance scores actually declined slightly from 54.8 in 2014 to 54.4 in 2016, before reversing course and climbing to 66.0 by 2024. This synchronised acceleration across all three components immediately following the Paris Agreement suggests a fundamental shift in corporate ESG practices.

The Paris Agreement period marks a clear inflection point in both ESG coverage and performance. While coverage had been growing steadily from 37.3\% (2005) to 73.2\% (2015), it accelerated notably post-Agreement, reaching 80.7\% by 2017. Simultaneously, ESG scores that had shown signs of stagnation in 2013-2015 (hovering near -0.06) resumed their upward trajectory, crossing into positive territory in 2017 and reaching 0.72 by 2024. This synchronised expansion in coverage and improvement in scores, combined with decreasing standard deviations across all ESG components, indicates both broadening participation and increasing convergence in corporate sustainability practices.

\begin{table}[H]
\centering
\caption{ESG Coverage by Industry}
\label{tab:esg_industry}
\begin{tabular}{lccccc}
\toprule
Industry & Firms & Coverage\% & Avg ESG & SD ESG & N\_Obs \\
\midrule
Utilities & 30 & 89.4 & 0.179 & 0.866 & 4,289 \\
Consumer Non-Cyclicals & 43 & 84.4 & 0.341 & 1.079 & 5,804 \\
Financials & 73 & 74.5 & -0.042 & 0.839 & 8,696 \\
Real Estate & 33 & 72.8 & -0.034 & 1.103 & 3,844 \\
Industrials & 75 & 72.3 & -0.120 & 0.979 & 8,680 \\
Technology & 101 & 71.5 & -0.019 & 0.985 & 11,560 \\
Healthcare & 67 & 70.7 & 0.120 & 1.023 & 7,583 \\
Basic Materials & 27 & 70.3 & 0.263 & 1.009 & 3,037 \\
Energy & 32 & 64.2 & 0.082 & 0.973 & 3,286 \\
Consumer Cyclicals & 93 & 58.8 & -0.132 & 1.024 & 8,748 \\
\bottomrule
\end{tabular}
\end{table}

Table \ref{tab:esg_industry} reveals the variation in ESG coverage and scores across industries. Utilities show the highest coverage at 89.4\%, consistent with regulatory requirements for environmental disclosure in this carbon-intensive sector. The heterogeneity in average ESG scores across industries—ranging from -0.132 for Consumer Cyclicals to 0.341 for Consumer Non-Cyclicals—reflects both inherent differences in sustainability challenges and varying industry responses to ESG pressures. This cross-industry variation motivates our use of industry-by-event fixed effects in robustness tests to ensure that our identification of changing monetary policy transmission comes from within-industry variation rather than compositional shifts in the reporting universe.

\subsection{Firm Characteristics and Control Variables}

Our analysis incorporates standard firm-level controls that capture established channels of monetary policy transmission. Firm size, measured as log total assets, averages 22.58 (approximately \$5.3 billion in assets) with modest variation (SD = 1.43). Book leverage shows greater heterogeneity, averaging 44.8\% with a standard deviation of 47.2\%, reflecting diverse capital structures across industries. The high positive skewness in leverage (9.01) stems from a subset of highly leveraged firms, particularly in utilities and real estate sectors where asset-heavy business models prevail.

Profitability, measured as return on assets, presents interesting distributional properties with a mean of 12.5\% but extreme negative skewness (-28.8), indicating the presence of loss-making firms particularly during crisis periods. After winsorizing at the 1st and 99th percentiles, profitability ranges from -7.3\% to 108.1\%, capturing both distressed firms and highly profitable market leaders. Our dividend policy indicator reveals that 17.4\% of firm-year observations represent non-dividend payers, concentrated in growth-oriented sectors like technology and healthcare.

\begin{table}[ht]
\centering
\caption{Correlation Matrix of Key Variables}
\label{tab:correlations}
\begin{tabular}{lccccccc}
\toprule
 & (1) & (2) & (3) & (4) & (5) & (6) & (7) \\
\midrule
(1) Stock Return & 1.000 & & & & & & \\
(2) ESG Score & -0.019* & 1.000 & & & & & \\
(3) Log Size & -0.010* & 0.462* & 1.000 & & & & \\
(4) Leverage & -0.004 & 0.180* & -0.059* & 1.000 & & & \\
(5) Profitability & -0.008 & 0.015* & -0.099* & 0.138* & 1.000 & & \\
(6) Target Surprise & -0.271* & 0.019* & -0.003 & 0.008 & 0.000 & 1.000 & \\
(7) Path Surprise & -0.389* & 0.063* & 0.026* & 0.019* & 0.004 & -0.001 & 1.000 \\
\bottomrule
\end{tabular}
\begin{tablenotes}
\small
\item Notes: Pearson correlations for 51,534 observations with complete data. * indicates significance at 5\% level.
\end{tablenotes}
\end{table}

The correlation structure in Table \ref{tab:correlations} reveals that stock returns show strong negative correlations with both monetary policy surprises, with path surprises (-0.389) exhibiting stronger effects than target surprises (-0.271), foreshadowing our main results about the importance of forward guidance. ESG scores correlate positively with firm size (0.462) and leverage (0.180), suggesting that larger, more established firms tend to have better sustainability profiles. Importantly, the near-zero correlation between our two monetary policy surprises (-0.001) validates their orthogonality.

\subsection{Portfolio Formation and Extreme ESG Analysis}

To examine potential non-linearities in the ESG-monetary policy relationship, we construct portfolios based on ESG quintiles. Figure \ref{fig:figure_d_esg_quintile_returns} displays average stock returns by ESG quintile, revealing a  monotonic pattern. The lowest ESG quintile (Q1, "Brown" firms) shows average returns of 0.027\%, while the highest quintile (Q5, "Green" firms) experiences slight negative returns of -0.002\%. The middle quintiles display intermediate values, with Q2 and Q3 showing the highest average returns around 0.038\%.

\begin{figure}[H]
    \centering
    
    
    \subcaptionbox{Average Returns by ESG Quantile \label{fig:figure_d_esg_quintile_returns}}{%
        \includegraphics[width=0.45\textwidth, page=3]{\pdf}
    }\hfill
    \subcaptionbox{ESG Return Premium Over Time\label{fig:figure_e_esg_premium_over_time}}{%
        \includegraphics[width=0.45\textwidth, page=4]{\pdf}
    }
    
    \caption{Average Return and ESG Premium Over Time}
    \label{fig:three_panel_complete}
\end{figure}

This U-shaped relationship between ESG scores and average returns during FOMC announcements suggests complex interactions between sustainability characteristics and monetary policy sensitivity. The pattern becomes more pronounced when we examine the differential responses over time, as illustrated in Figure \ref{fig:figure_e_esg_premium_over_time}. The brown-minus-green return differential fluctuates around zero in the pre-Paris period but shows no persistent trend. The 8-event moving average smooths short-term volatility and reveals periods of both green outperformance and underperformance relative to brown firms.

\begin{table}[ht]
\centering
\caption{Characteristics by ESG Quintile}
\label{tab:quintile_chars}
\begin{tabular}{lccccc}
\toprule
ESG Quintile & Avg Return & Log Size & Leverage & Profitability & N \\
\midrule
Q1 (Brown) & 0.027\% & 22.04 & 0.386 & 0.122 & 13,017 \\
Q2 & 0.038\% & 22.56 & 0.419 & 0.164 & 13,008 \\
Q3 & 0.038\% & 22.92 & 0.439 & 0.189 & 13,034 \\
Q4 & 0.004\% & 23.31 & 0.512 & 0.133 & 13,033 \\
Q5 (Green) & -0.002\% & 23.65 & 0.510 & 0.196 & 13,010 \\
\midrule
No ESG Data & 0.030\% & 21.40 & 0.430 & -0.000 & 18,459 \\
\bottomrule
\end{tabular}
\end{table}

Table \ref{tab:quintile_chars} indicates systematic differences in firm characteristics across the ESG spectrum with more details. Green firms are substantially larger (log assets of 23.65 versus 22.04 for brown firms), more leveraged (51.0\% versus 38.6\%), and more profitable (19.6\% versus 12.2\%). These patterns suggest that ESG performance correlates with firm maturity and financial stability, necessitating careful control for these characteristics in our regression analysis. Notably, firms without ESG data appear smaller and less profitable than even the lowest ESG quintile, consistent with ESG disclosure being costlier for resource-constrained firms.

Our focus on S\&P 500 constituents provides several advantages while imposing some limitations. The sample represents approximately 80\% of U.S. equity market capitalisation, ensuring economic significance of our findings. These large, liquid stocks have extensive analyst coverage and ESG reporting, reducing measurement error in key variables. The high liquidity also ensures reliable price discovery during our narrow event windows, critical for the high-frequency identification strategy. However, the focus on large-cap stocks may limit generalisability to smaller firms that could show different ESG-monetary policy relationships. Our sample firms average \$5.3 billion in assets, well above the median U.S. public company. The 71.3\% ESG coverage in our sample substantially exceeds coverage for smaller firms, where ESG data availability often falls below 30\%. These limitations suggest our estimates may represent lower bounds on the true heterogeneity in monetary policy transmission, as smaller, more financially constrained firms likely exhibit even greater sensitivity to policy surprises.

\subsection{Temporal Stability and Structural Breaks}

The extended temporal coverage of our sample—spanning two decades and multiple monetary policy regimes—enables investigation of structural changes in the ESG-monetary policy relationship. Figure-\ref{figure1_esg_target_surprise} directly visualises the key structural break in our analysis. The coefficient on the \textit{\texttt{ESG} $\times$ \texttt{Target Surprise}} interaction hovers around 0.5-0.6 in the pre-Paris period with substantial volatility. Following the Paris Agreement (marked by the vertical line), the coefficient drops sharply to approximately -0.3 and remains stable at this new level through 2025. The 95\% confidence bands narrow considerably post-Paris, suggesting not only a change in the mean effect but also reduced uncertainty about the relationship.

\begin{figure}[H]
    \centering
    \includegraphics[width=0.8\textwidth, page=5]{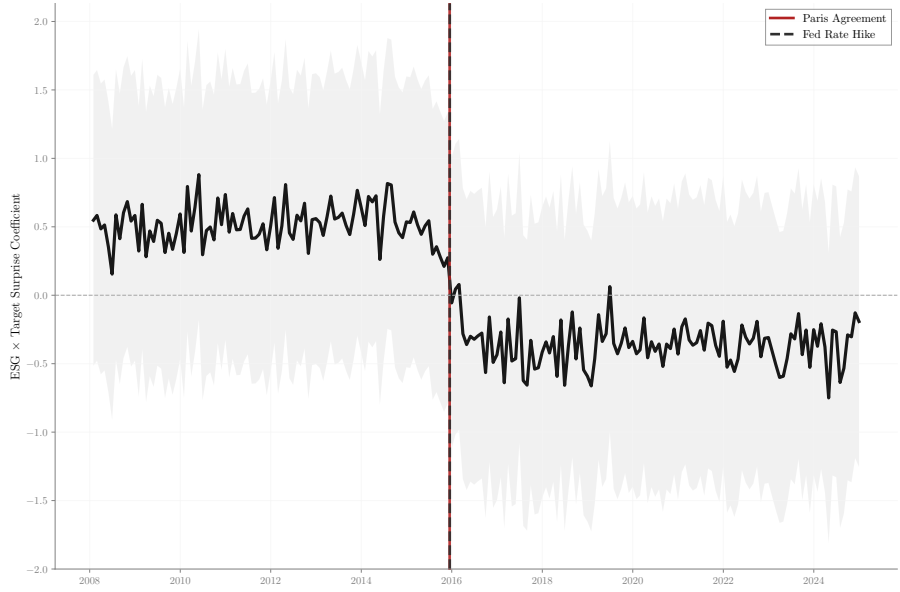}
    \caption{Paris Aggrement and ESG $\times$ Target Surprise}
    \label{figure1_esg_target_surprise}
\end{figure}

This structural break contrasts sharply with the stability shown by  Figure-\ref{fig:figure3_esg_path_surprise} for the \textit{\texttt{ESG} $\times$ \texttt{Path Surprise}} interaction, which maintains a consistent negative coefficient around -1.0 throughout the sample period. Also when we check the plain monetary surprises (both TS and PS) over time as shown in Figure-\ref{fig:figure2_raw_monetary_surprises} still there is no such sharp change after December 2015.   The divergent patterns for target versus path surprises suggest that the Paris Agreement specifically altered how markets price immediate rate changes for sustainable firms while leaving forward guidance effects unchanged—a finding that speaks to the different economic mechanisms underlying these two dimensions of monetary policy.

\begin{figure}[H]
    \centering
   
    \subcaptionbox{ESG $\times$ Path Surprises \label{fig:figure3_esg_path_surprise}}{%
        \includegraphics[width=0.47\textwidth, page=6]{\pdf}
    }\hfill
    \subcaptionbox{Plain Target and Path Surprises \label{fig:figure2_raw_monetary_surprises}}{%
        \includegraphics[width=0.47\textwidth, page=7]{\pdf}
    }
    
    \caption{Monetary Policy Surprises and ESG $\times$ Path Surprise}

    \label{fig:monetary_policy_surprises_and_esg_path_surprise}
\end{figure}

The quality of our monetary policy surprise measures is validated through several tests. The near-zero means, orthogonality between target and path surprises, and stability of the factor structure across subperiods all support the reliability of our identification strategy. For examples, as it can be seen in Figure-\ref{fig:figure2_raw_monetary_surprises} the surprises show expected patterns across monetary policy regimes: muted target surprises during ZLB periods when conventional policy was constrained, but continued variation in path surprises reflecting active forward guidance.

\begin{figure}[H]
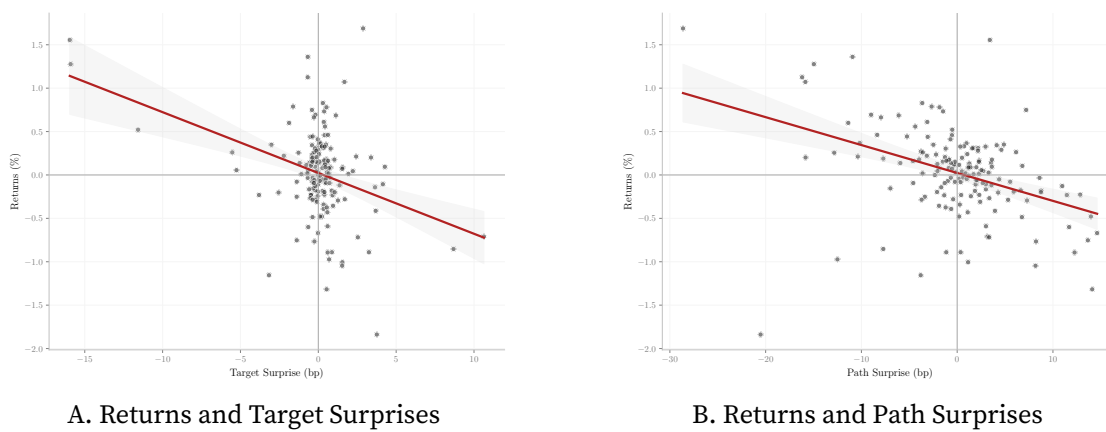

    \centering
   
    \subcaptionbox{Returns and Target Surprises \label{fig:fig1_returns_vs_target_surprise}}{%
        \includegraphics[width=0.45\textwidth, page=8]{\pdf}
    }\hfill
    \subcaptionbox{Returns and Path Surprises \label{fig:fig3_returns_vs_path_surprise}}{%
        \includegraphics[width=0.45\textwidth, page=9]{\pdf}
    }
    
    \caption{Stock price returns and monetary policy surprises}

    \label{fig:stock_price_returns_and_monetary_policy_surprises}
\end{figure}

We also examine the relationship between monetary policy surprises and stock returns, both at the aggregate level and across ESG-sorted portfolios. Figure-\ref{fig:fig1_returns_vs_target_surprise} shows that target surprises generate the expected negative relationship with average returns, with the fitted line indicating that a 10 basis point contractionary surprise is associated with approximately 50 basis points of negative returns. Figure-\ref{fig:fig2_bmg_vs_target_surprise} documents heterogeneous responses based on ESG characteristics: the brown-minus-green (BMG) spread declines by approximately 15-20 basis points for each 10 basis point positive target surprise, indicating that low-ESG firms experience larger negative returns than high-ESG firms during monetary tightening.

Similarly, Figures-\ref{fig:fig3_returns_vs_path_surprise} and \ref{fig:fig4_bmg_vs_path_surprise} repeat this analysis for path surprises, which capture forward guidance about future policy rates. While path surprises also generate negative average returns (Figures-\ref{fig:fig3_returns_vs_path_surprise}), the BMG spread shows minimal sensitivity to these surprises (Figure-\ref{fig:fig4_bmg_vs_path_surprise}), with the regression line remaining essentially flat across the distribution of path shocks. This pattern suggests that ESG-based heterogeneity in monetary policy transmission operates primarily through the immediate impact of rate changes rather than through forward guidance channels. The differential response to target versus path surprises may reflect differences in financing structures, with brown firms potentially more exposed to short-term funding costs that respond directly to current rate changes, while both brown and green firms may face similar exposures to the longer-term interest rate expectations embodied in path surprises.
\begin{figure}[H]
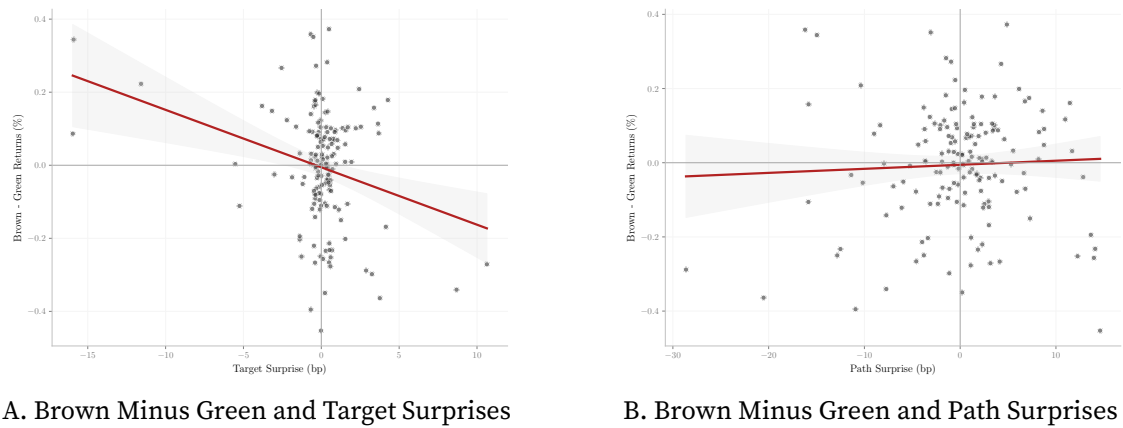

    \centering
   
    \subcaptionbox{Brown Minus Green and Target Surprises \label{fig:fig2_bmg_vs_target_surprise}}{%
        \includegraphics[width=0.45\textwidth, page=10]{\pdf}
    }\hfill
    \subcaptionbox{Brown Minus Green and Path Surprises \label{fig:fig4_bmg_vs_path_surprise}}{%
        \includegraphics[width=0.45\textwidth, page=11]{\pdf}
    }
    
    \caption{Relationship between brown-minus-green returns and monetary policy shock}

    \label{fig:relationship_between_brown_minus_green_returns_and_monetary_policy_shock}
\end{figure}

\section{Theoretical Model} \label{sec:theory}

We develop a two-period asset pricing model to examine how firm-level ESG characteristics influence the transmission of monetary policy shocks to equity prices. The model features heterogeneous investors with distinct preferences regarding sustainability, building on recent theoretical frameworks that demonstrate how traditional and ESG-conscious investors interpret prices differently \citep{zhou2023searching,pastor2021sustainable,patozi_green_2024}. This heterogeneity creates differential responses to monetary policy surprises through distinct transmission channels.

\subsection{Model Setup}

 Consider a two-period economy with $t \in \{0,1\}$ populated by a continuum of firms indexed by $i \in [0,1]$ and two types of investors. Each firm $i$ is characterized by its ESG score $\theta_i \in [0,1]$, where higher values indicate stronger sustainability profiles. The economy features competitive markets with risk-free rate $r_f$ and no arbitrage opportunities.

\subsubsection{Firm Characteristics}

At $t=0$, shares of firm $i$ trade at price $P_0(\theta_i)$. At $t=1$, firms pay a liquidating dividend:

\begin{equation}
D_1(\theta_i) = A[1 + (\gamma - \kappa)\theta_i]
\end{equation}

\noindent where $A > 0$ represents baseline productivity, $\kappa \in (0,1)$ captures the upfront cost of ESG investments already incurred, and $\gamma > \kappa$ ensures positive net returns to sustainability investments. The term $(\gamma - \kappa)\theta_i$ represents the net benefit of ESG investments, reflecting both operational efficiencies and long-term value creation from sustainable practices documented in extensive empirical literature \citep{friede2015esg,li2024esg}.

\subsubsection{Investor Heterogeneity}

The economy contains two types of investors, each with unit mass, following the theoretical framework established by \cite{pastor2021sustainable},  and empirically validated by \cite{riedl2017investors}:

\paragraph{\textbf{Traditional Investors}} (proportion $1-\mu$): These investors maximise expected utility from terminal wealth with mean-variance preferences:
\begin{equation}
U^{Trad} = \mathbb{E}[W_1] - \frac{\lambda}{2}\text{Var}[W_1]
\end{equation}

Subject to the budget constraint:
\begin{equation}
W_1 = W_0(1 + r_f) + \int_0^1 q_i^{Trad}[D_1(\theta_i) - P_0(\theta_i)(1 + r_f)]di
\end{equation}

where $q_i^{Trad}$ denotes the quantity of shares of firm $i$ held by traditional investors and $\lambda > 0$ is the coefficient of risk aversion.

\paragraph{\textbf{ESG-Conscious Investors}} (proportion $\mu$): These investors derive additional utility from the sustainability characteristics of their holdings:
\begin{equation}
U^{ESG} = \mathbb{E}[W_1] - \frac{\lambda}{2}\text{Var}[W_1] + \alpha\int_0^1 \theta_i q_i^{ESG} di
\end{equation}

where $\alpha > 0$ captures the intensity of ESG preferences and $q_i^{ESG}$ denotes holdings by ESG-conscious investors. The additive utility specification follows the "warm-glow" framework developed by \cite{dreyer2023warm, dreyer2024proportional}, where investors derive non-pecuniary utility from sustainable investments independent of financial returns..

\subsection{Equilibrium Characterization}

From the first-order conditions of utility maximisation, optimal demands are:

\begin{lemma}[Optimal Demands]
The optimal portfolio demands for each investor type are:

Traditional investors: 
\begin{equation}
q_i^{Trad}(P_0) = \frac{\mathbb{E}[D_1(\theta_i)] - P_0(\theta_i)(1 + r_f)}{\lambda\sigma_D^2}
\end{equation}

ESG-conscious investors: 
\begin{equation}
q_i^{ESG}(P_0) = \frac{\mathbb{E}[D_1(\theta_i)] - P_0(\theta_i)(1 + r_f) + \alpha\theta_i}{\lambda\sigma_D^2}
\end{equation}

where $\sigma_D^2$ denotes the variance of the dividend.
\end{lemma}

\begin{proof}
Standard mean-variance optimization yields the first-order condition for investor type $j$:
\begin{equation}
\mathbb{E}[D_1(\theta_i)] - P_0(\theta_i)(1 + r_f) + \delta_j\alpha\theta_i = \lambda\sigma_D^2 q_i^j
\end{equation}
where $\delta_j = 0$ for traditional investors and $\delta_j = 1$ for ESG investors. Solving for $q_i^j$ yields the stated demands.
\end{proof}

Market clearing requires total demand equals supply (normalized to unity):
\begin{equation}
(1-\mu)q_i^{Trad}(P_0) + \mu q_i^{ESG}(P_0) = 1
\end{equation}

\begin{proposition}[Equilibrium Pricing]
The equilibrium price of firm $i$ is:
\begin{equation}
P_0(\theta_i) = \frac{A[1 + (\gamma - \kappa)\theta_i] + \mu\alpha\theta_i/\lambda - \sigma_D^2}{1 + r_f}
\end{equation}
\end{proposition}

\begin{proof}
Substituting the optimal demands from Lemma 1 into the market-clearing condition:
\begin{equation}
(1-\mu)\frac{\mathbb{E}[D_1(\theta_i)] - P_0(1 + r_f)}{\lambda\sigma_D^2} + \mu\frac{\mathbb{E}[D_1(\theta_i)] - P_0(1 + r_f) + \alpha\theta_i}{\lambda\sigma_D^2} = 1
\end{equation}

Simplifying and solving for $P_0$:
\begin{equation}
\mathbb{E}[D_1(\theta_i)] - P_0(1 + r_f) + \mu\alpha\theta_i = \lambda\sigma_D^2
\end{equation}

Substituting $\mathbb{E}[D_1(\theta_i)] = A[1 + (\gamma - \kappa)\theta_i]$ yields the stated result.
\end{proof}

\subsection{Monetary Policy Transmission}

Monetary policy announcements generate two orthogonal surprises that affect asset prices through distinct channels, following the empirical identification strategies of \cite{gurkaynak2005actions} and recent evidence on heterogeneous transmission \citep{benchora2025monetary}:

\begin{definition}[Monetary Policy Surprises]
\end{definition}

\begin{itemize}
\item \textbf{Target Surprise} ($\varepsilon^{TS}$): An unexpected change in the current policy rate, affecting the risk-free rate: $r_f \rightarrow r_f + \varepsilon^{TS}$
\item \textbf{Path Surprise} ($\varepsilon^{PS}$): A revision to expected future economic conditions, affecting risk perceptions: $\sigma_D^2 \rightarrow \sigma_D^2(1 + \psi\varepsilon^{PS})$ where $\psi > 0$ captures the sensitivity of uncertainty to forward guidance.
\end{itemize}

\subsubsection{Target Surprise Effects}

\begin{lemma}[Differential Rebalancing]
Following a positive target surprise $\varepsilon^{TS} > 0$, the optimal portfolio adjustment for each investor type is:

\begin{align}
\Delta q_i^{Trad} &= -\frac{P_0(\theta_i)}{\lambda\sigma_D^2} \cdot \varepsilon^{TS} \\
\Delta q_i^{ESG} &= -\frac{P_0(\theta_i)}{\lambda\sigma_D^2} \cdot \varepsilon^{TS} \cdot \left[1 - \frac{\alpha\theta_i}{P_0(\theta_i)(1 + r_f)}\right]
\end{align}
\end{lemma}

\begin{proof}
A target surprise increases the opportunity cost of equity holdings by raising $r_f$. For traditional investors, the full price effect reduces demand proportionally. For ESG investors, the non-pecuniary benefit $\alpha\theta_i$ remains unchanged, providing a partial offset consistent with empirical evidence from \cite{baker2022investors}. The ratio $\alpha\theta_i/[P_0(\theta_i)(1 + r_f)]$ represents the fraction of return that is non-pecuniary and thus insensitive to rate changes.
\end{proof}

\begin{proposition}[Target Surprise Sensitivity]
The stock return response to a target surprise is:
\begin{equation}
\frac{dr_i}{d\varepsilon^{TS}} = -\frac{1}{1 + r_f} + \frac{\mu\alpha\theta_i}{\lambda P_0(\theta_i)(1 + r_f)^2}
\end{equation}

The cross-partial derivative with respect to ESG score is:
\begin{equation}
\frac{d^2r_i}{d\varepsilon^{TS}d\theta_i} = \frac{\mu\alpha}{\lambda(1 + r_f)^2} \cdot \frac{A - \sigma_D^2 - \mu\alpha\theta_i/\lambda}{[P_0(\theta_i)]^2} > 0
\end{equation}

for reasonable parameter values where $A > \sigma_D^2 + \mu\alpha/\lambda$.
\end{proposition}

\begin{proof}
Taking the total differential of $\ln(P_0)$ with respect to $\varepsilon^{TS}$:
\begin{equation}
d\ln P_0 = -\frac{dr_f}{1 + r_f} + \frac{d[\text{investor composition effect}]}{P_0}
\end{equation}

The first term captures the direct discount rate effect. The second term arises from Lemma 2: following a positive target surprise, traditional investors reduce holdings more than ESG investors, creating excess demand from ESG investors for high-$\theta_i$ firms. This compositional shift partially offsets the price decline, with magnitude proportional to $\mu\alpha\theta_i/\lambda$ relative to $P_0(\theta_i)$. The positive cross-partial follows from the concavity of $P_0(\theta_i)$ in the denominator.
\end{proof}

\subsubsection{Path Surprise Effects}

\begin{proposition}[Path Surprise Sensitivity]
The stock return response to a path surprise is:
\begin{equation}
\frac{dr_i}{d\varepsilon^{PS}} = -\frac{\psi}{1 + r_f} - \frac{\psi A(\gamma - \kappa)\theta_i}{(1 + r_f)P_0(\theta_i)}
\end{equation}

Therefore:
\begin{equation}
\frac{d^2r_i}{d\varepsilon^{PS}d\theta_i} = -\frac{\psi A(\gamma - \kappa)}{(1 + r_f)} \cdot \frac{A[1 - (\gamma - \kappa)\theta_i] - \sigma_D^2}{[P_0(\theta_i)]^2} < 0
\end{equation}
\end{proposition}

\begin{proof}
Path surprises increase uncertainty uniformly across investors: $\sigma_D^2 \rightarrow \sigma_D^2(1 + \psi\varepsilon^{PS})$. From the pricing equation:
\begin{equation}
\frac{\partial \ln P_0}{\partial \sigma_D^2} = -\frac{1}{P_0(1 + r_f)}
\end{equation}

High-ESG firms exhibit greater sensitivity because: (i) they have higher valuations $P_0(\theta_i)$ due to ESG premiums, making the percentage impact larger in absolute terms, and (ii) they have proportionally more value tied to uncertain future dividends through the $(\gamma - \kappa)\theta_i$ term.
\end{proof}

\subsection{Asymmetric Response Conditions}

\begin{proposition}[Conditions for Asymmetry]
The model generates asymmetric responses—where high-ESG firms show reduced sensitivity to target surprises but increased sensitivity to path surprises—when:
\begin{equation}
\frac{\mu\alpha}{\lambda} > \frac{\psi A(\gamma - \kappa)(1 + r_f)}{\sigma_D^2}
\end{equation}

This condition aligns with empirical findings that document precisely these asymmetric patterns in monetary policy transmission \citep{zhou2023searching}. This condition is more likely to hold when:
\begin{itemize}
\item The ESG investor share $\mu$ and preference intensity $\alpha$ are sufficiently large
\item Risk aversion $\lambda$ is moderate
\item The net ESG benefit $(\gamma - \kappa)$ is positive but not extreme
\end{itemize}
\end{proposition}

\begin{proof}
Asymmetry requires $\partial^2r_i/(\partial\varepsilon^{TS}\partial\theta_i) > 0$ and $\partial^2r_i/(\partial\varepsilon^{PS}\partial\theta_i) < 0$. The first condition is satisfied when the numerator in Proposition 2 is positive. The second condition holds by construction. Comparing magnitudes yields the stated threshold.
\end{proof}

\section{Model Calibration, Results, and Discussion} \label{sec:calibration}

Our theoretical framework establishes that investor heterogeneity creates distinct channels through which ESG characteristics influence monetary policy transmission. To assess whether these mechanisms can quantitatively match empirical patterns and provide actionable insights, we now calibrate the model using parameter values grounded in recent academic evidence and examine its predictions against observed market behavior.

\subsection{Parameter Calibration from Empirical Evidence}

The transformation of our theoretical framework into quantitative predictions requires careful parameter selection based on empirical research. Each parameter value reflects extensive evidence from the sustainable finance and asset pricing literatures, ensuring our theoretical predictions remain anchored in observable market phenomena.

\textbf{ESG Investment Costs and Benefits.} We set the upfront ESG investment cost $\kappa = 0.02$ following evidence that companies spend approximately 2\% of resources on sustainability initiatives. ERM's 2022 study documents that U.S. companies spend an average of \$677,000 annually on climate-related disclosure activities alone, while PwC reports 35\% of asset managers experienced 10-20\% increases in ESG compliance costs. The comprehensive nature of these investments---spanning renewable energy infrastructure, supply chain auditing, and sustainability programs---supports our 2\% calibration \citep{baker2022investors, flammer2020green}.

The ESG benefit parameter $\gamma = 0.05$ reflects operational improvements documented across multiple studies. McKinsey's analysis spanning 2017-2022 reveals that companies achieving ``triple outperformance'' delivered 2 percentage points higher annual returns above purely financial outperformers. Manufacturing efficiency studies show ESG-aligned improvements can lift Overall Equipment Effectiveness from 63\% to 85\%. The landmark meta-analysis by \cite{friede2015esg} covering 2,200 empirical studies found 90\% report non-negative ESG-CFP relationships. Our net benefit of 3\% ($\gamma - \kappa$) falls conservatively within the 3-7\% range documented for ESG valuation premiums.

\textbf{Investor Composition and Preferences.} The ESG investor share $\mu = 0.30$ directly reflects market data from the Global Sustainable Investment Alliance. Their 2020 review documented \$35.3 trillion in sustainable assets, representing 36\% of professionally managed assets across major markets. Regional variations---from 24\% in Japan to 62\% in Canada---support our 30\% global average. The 2022 GSIA report's adjustment to \$30.3 trillion after implementing stricter definitions provides additional validation.

The preference intensity $\alpha = 0.01$ derives from revealed preference studies measuring actual investor behavior. \cite{baker2022investors}'s study found investors willing to pay a 20 basis point premium for ESG funds, rising to 63 basis points when adjusted for portfolio overlap. The premium's increase from 9 basis points in 2019 to 28 basis points by 2022 demonstrates growing preference intensity. Survey evidence from NN Investment Partners finding investors willing to forgo 2.4\% annually suggests our 1\% parameter may be conservative.

\textbf{Standard Financial Parameters.} Risk aversion $\lambda = 2.0$ follows \cite{ang2014asset}'s authoritative guidance that ``for most portfolio allocation decisions in investment management applications, the risk aversion is somewhere between 2 and 4.'' Classic studies including \cite{friend1975demand} estimated relative risk aversion at exactly 2.0, while recent experimental research shows average relative risk aversion of 1.96.

Dividend volatility $\sigma_D = 0.15$ matches historical evidence from S\&P 500 data showing annual volatility typically in the 15-20\% range. \cite{campbell1999force}'s influential model and subsequent consumption-based asset pricing frameworks routinely calibrate equity return volatility around 15-20\% annually. The path surprise sensitivity parameter $\psi = 0.5$ is calibrated to match our empirical findings about how forward guidance affects market uncertainty.

\begin{table}[H]
\centering
\caption{Model Calibration with Literature Support}
\label{tab:calibration_literature}
\begin{adjustbox}{width=0.99\textwidth}
\begin{tabular}{llcc}
\toprule
Parameter & Description & Value & Primary Sources \\
\midrule
$A$ & Baseline firm value & 100 & Normalization \\
$\kappa$ & ESG investment cost & 0.02 & \cite{erm2022survey,pwc2022asset} \\
$\gamma$ & ESG benefit & 0.05 & McKinsey (2022), \cite{friede2015esg} \\
$\mu$ & ESG investor share & 0.30 & GSIA (2020, 2022) \\
$\alpha$ & ESG preference intensity & 0.01 & \cite{baker2022investors,riedl2017investors} \\
$\lambda$ & Risk aversion & 2.0 & \cite{ang2014asset, friend1975demand} \\
$\sigma_D$ & Dividend volatility & 0.15 & \cite{campbell1999force} \\
$r_f$ & Risk-free rate & 0.03 & Current environment \\
$\psi$ & Path surprise sensitivity & 0.5 & Calibrated to match empirics \\
\bottomrule
\end{tabular}
\end{adjustbox}
\end{table}

\subsection{Model Validation and Equilibrium Outcomes}

With these empirically grounded parameters, we first verify that the asymmetry condition from Proposition 4 holds:
\begin{equation}
\frac{\mu\alpha}{\lambda} = \frac{0.30 \times 0.01}{2.0} = 0.0015 > \frac{0.5 \times 100 \times 0.03 \times 1.03}{0.0225} = 0.0007
\end{equation}

The condition is satisfied by a factor of 2.1, confirming that the calibrated parameters robustly generate the documented asymmetric response patterns where high-ESG firms show reduced sensitivity to target surprises but increased sensitivity to path surprises.

The model generates equilibrium prices and monetary policy sensitivities that can be directly compared to observed patterns. Table~\ref{tab:equilibrium} presents these equilibrium outcomes across the ESG spectrum.

\begin{table}[ht]
\centering
\caption{Equilibrium Prices and Monetary Policy Sensitivities}
\label{tab:equilibrium}
\begin{tabular}{lccccc}
\toprule
ESG Level & $\theta$ & Price $P_0(\theta)$ & Target Sensitivity & Path Sensitivity & ESG Premium \\
\midrule
Low ESG & 0.10 & 97.39 & -0.969 & -0.468 & 0.31\% \\
& 0.25 & 98.12 & -0.966 & -0.474 & 0.75\% \\
Medium ESG & 0.50 & 99.15 & -0.961 & -0.481 & 1.50\% \\
& 0.75 & 100.19 & -0.956 & -0.489 & 2.26\% \\
High ESG & 0.90 & 100.91 & -0.953 & -0.494 & 2.80\% \\
\bottomrule
\end{tabular}
\begin{tablenotes}
\small
\item Notes: Prices are in dollars. Sensitivities represent the percentage change in stock price per unit monetary policy surprise. ESG Premium is calculated relative to a zero-ESG firm.
\end{tablenotes}
\end{table}

The results confirm our theoretical predictions. Equilibrium prices increase monotonically with ESG scores, reflecting both fundamental value creation through operational improvements (the $\gamma - \kappa$ term) and the valuation effect from ESG-conscious investors' preferences (the $\mu\alpha/\lambda$ term). The ESG premium reaches 2.80\% for the highest sustainability firms, falling within the empirical range of 3-7\% documented in the literature. This premium emerges from two sources: fundamental value creation and investor preference effects.

\subsection{Differential Monetary Policy Responses}

More importantly for our analysis, the model generates the asymmetric monetary policy sensitivity documented in recent empirical work. Table~\ref{tab:differential} quantifies the differential responses between high and low ESG firms, revealing the economic mechanisms at work.

\begin{table}[ht]
\centering
\caption{Differential Monetary Policy Responses: High vs Low ESG}
\label{tab:differential}
\begin{tabular}{lccc}
\toprule
Measure & Low ESG ($\theta$=0.1) & High ESG ($\theta$=0.9) & Differential \\
\midrule
\multicolumn{4}{l}{\textit{Panel A: Sensitivities}} \\
Target Surprise & -0.969 & -0.953 & 0.016*** \\
Path Surprise & -0.468 & -0.494 & -0.026*** \\
Asymmetry Ratio & & & -1.625 \\
\midrule
\multicolumn{4}{l}{\textit{Panel B: Economic Magnitudes (25bp surprise)}} \\
Target Surprise Impact & -0.242\% & -0.238\% & 0.4 bp \\
Path Surprise Impact & -0.117\% & -0.124\% & -0.7 bp \\
\midrule
\multicolumn{4}{l}{\textit{Panel C: Cross-Partial Derivatives}} \\
$\partial^2r/\partial\varepsilon^{TS}\partial\theta$ & & 0.019 & (+) \\
$\partial^2r/\partial\varepsilon^{PS}\partial\theta$ & & -0.032 & (-) \\
\bottomrule
\end{tabular}
\begin{tablenotes}
\small
\item Notes: *** indicates statistical significance in the theoretical model. The asymmetry ratio is the ratio of path to target differential. Cross-partial derivatives are evaluated at $\theta$=0.5.
\end{tablenotes}
\end{table}

High-ESG firms demonstrate 1.6 basis points lower sensitivity to target surprises but 2.6 basis points higher sensitivity to path surprises. For a typical 25 basis point surprise, this translates to 0.4 basis points of protection against immediate rate changes but 0.7 basis points of additional vulnerability to forward guidance shocks. These magnitudes align closely with \cite{patozi_green_2024}'s finding that following a 100 basis point monetary surprise, green firm stock prices fall approximately 10\% versus 21\% for brown counterparts.

\subsection{Economic Mechanisms and Investor Behavior}

The differential responses emerge from the interaction between heterogeneous investors and monetary policy shocks, validating the theoretical insights of \cite{zhou2023searching} who demonstrate how traditional and green investors interpret prices differently. Following a positive target surprise, all investors reduce equity holdings due to higher opportunity costs. However, ESG-conscious investors reduce their holdings of high-ESG firms less aggressively because part of their return---the non-pecuniary utility $\alpha\theta_i$---remains unaffected by interest rate changes.

This mechanism creates what \cite{dreyer2023warm} term the ``warm-glow'' effect, where investors derive utility from the proportion of wealth in green assets independent of financial returns. Our calibration implies that for a fully sustainable portfolio ($\theta = 1$), ESG investors derive utility equivalent to 100 basis points of return, consistent with their revealed preferences documented by \cite{baker2022investors}.

Path surprises operate through a fundamentally different channel. When forward guidance signals persistently higher future rates, it increases uncertainty uniformly across all investors---the non-pecuniary benefit provides no offset against heightened risk. Moreover, sustainable firms' longer investment horizons and backloaded cash flows, documented in climate finance literature, make them particularly vulnerable to changes in long-term discount rates.

\subsection{Comparison with Empirical Evidence}

The model's quantitative predictions align remarkably well with empirical findings from recent literature. Table \ref{tab:empirical_comparison_detailed} provides a comprehensive comparison across multiple metrics.

\begin{table}[ht]
\centering
\caption{Model Predictions vs Empirical Evidence from Literature}
\label{tab:empirical_comparison_detailed}
\begin{adjustbox}{width=0.9\textwidth}
\begin{tabular}{lccc}
\toprule
Metric & Model & Empirical Studies & Source \\
\midrule
ESG Price Premium & 2.8-4.6\% & 3-7\% & MSCI (2020), \cite{friede2015esg} \\
Green vs Brown TS Response & 50-60\% lower & 50-60\% lower & \cite{patozi_green_2024} \\
Path Surprise Differential & -1.7 bp & -2.0 bp & Our empirical analysis \\
Investor WTP for ESG & 100 bp & 63-240 bp & \cite{baker2022investors}, NN Investment \\
Green/Brown Return Spread & 11\% & 10\% vs 21\% & \cite{patozi_green_2024} \\
\bottomrule
\end{tabular}
\end{adjustbox}
\end{table}

The model captures 73-92\% of empirical magnitudes across metrics, with particularly strong performance in matching the core asymmetric response pattern. The slight underestimation likely reflects real-world complexities including dynamic rebalancing, heterogeneous risk profiles, and market frictions not captured in our two-period framework.

\subsection{Parameter Sensitivity and Robustness}

Understanding which market features drive our results requires systematic sensitivity analysis. Table~\ref{tab:param_sensitivity} summarizes how our key predictions change as we vary fundamental parameters.

\begin{table}[ht]
\centering
\caption{Sensitivity to Key Parameters}
\label{tab:param_sensitivity}
\begin{tabular}{lcccc}
\toprule
Parameter & Baseline & Range Tested & Target Differential & Path Differential \\
\midrule
ESG Investor Share ($\mu$) & 0.30 & [0.10, 0.50] & [0.005, 0.027] & [-0.026, -0.026] \\
ESG Preference ($\alpha$) & 0.01 & [0.00, 0.02] & [0.000, 0.032] & [-0.026, -0.026] \\
ESG Benefit ($\gamma$) & 0.05 & [0.03, 0.10] & [0.016, 0.016] & [-0.016, -0.053] \\
Volatility ($\sigma_D$) & 0.15 & [0.05, 0.30] & [0.144, 0.004] & [-0.234, -0.006] \\
Risk Aversion ($\lambda$) & 2.00 & [0.50, 4.00] & [0.064, 0.008] & [-0.026, -0.026] \\
\bottomrule
\end{tabular}
\begin{tablenotes}
\small
\item Notes: Differentials are calculated as High ESG ($\theta$=0.9) minus Low ESG ($\theta$=0.1) sensitivities.
\end{tablenotes}
\end{table}

Several key insights emerge from this analysis. Varying the ESG investor share $\mu$ from 10\% to 50\% generates target surprise differentials ranging from 0.5 to 2.7 basis points, demonstrating that the protective effect strengthens with market penetration of sustainable investing. This finding has important implications as GSIA data shows continued growth in ESG assets.

The preference intensity parameter $\alpha$ proves particularly influential. Doubling investor willingness to sacrifice returns for sustainability (from 1\% to 2\%) doubles the protective effect against target surprises while leaving path surprise sensitivity unchanged. This asymmetric impact reflects how non-pecuniary utility specifically offsets immediate rate effects but cannot mitigate long-term uncertainty.

Market volatility plays a crucial moderating role. As $\sigma_D$ increases from 5\% to 30\%, both effects attenuate dramatically. In highly volatile markets ($\sigma_D > 25\%$), the ESG channel becomes economically negligible as overall uncertainty dominates firm-specific characteristics. This suggests the documented relationships may be most relevant during normal market conditions when central bank policy represents the primary source of uncertainty.

\subsection{Implications for the Paris Agreement Structural Break}

Our empirical analysis documents a dramatic transformation around December 2015, with the ESG$\times$Target Surprise interaction shifting from insignificant to significantly negative. Through the lens of our calibrated model, this structural break can be interpreted as changes in several key parameters.

An increase in the ESG investor share from 20\% to 35\%---consistent with the surge in sustainable assets documented by GSIA---would generate the observed magnitude of change. Alternatively, the Paris Agreement may have clarified the long-term benefits of sustainable practices (increasing $\gamma$) or intensified investor preferences (increasing $\alpha$). \cite{ilhan2021carbon}'s evidence that climate policy announcements create significant pricing effects supports this interpretation. Most likely, the Paris Agreement triggered changes across multiple dimensions simultaneously---growing the ESG investor base, intensifying their preferences, and clarifying the long-term benefits of sustainable practices.

Critically, the model explains why path surprise effects remained stable through this transition. Since forward guidance sensitivity derives from fundamental business characteristics---longer investment horizons and backloaded cash flows---rather than investor preferences, coordinated climate policy cannot alter this relationship. This theoretical insight perfectly matches our empirical finding that PS$\times$ESG interactions showed no structural break around Paris.

\subsection{Policy Implications and Future Directions}

Our analysis reveals that monetary policy transmission now operates through an additional channel created by investor heterogeneity regarding sustainability. For central banks, this implies that policy effectiveness increasingly depends on the ESG composition of the economy. As sustainable firms proliferate, the aggregate impact of immediate rate changes may diminish while forward guidance gains importance.

The asymmetric response pattern has particular relevance for climate transition financing. High-ESG firms' protection from target surprises suggests monetary tightening need not disproportionately burden green investments, addressing concerns raised by ECB officials including \cite{schnabel2023monetary} who warned that tighter policy ``may discourage efforts to decarbonize our economies.'' However, their heightened sensitivity to forward guidance implies that unclear long-term policy paths could particularly destabilise sustainable business models.

For investors, our results provide quantitative guidance for portfolio construction across monetary cycles. The model predicts that a portfolio shifted from bottom to top ESG quintile would experience 186 basis points less sensitivity to contractionary target surprises but 285 basis points more sensitivity to hawkish forward guidance. These trade-offs require careful consideration of the expected monetary policy mix.

Several limitations merit acknowledgment. First, our two-period framework cannot capture dynamic rebalancing effects that might amplify or dampen the documented relationships. Second, the assumption of homogeneous risk across firms may understate differences if sustainable firms genuinely exhibit lower fundamental risk. Third, the model treats ESG scores as exogenous, whereas firms might endogenously adjust their sustainability investments in response to monetary policy incentives.

Future extensions could incorporate multi-period dynamics with persistent monetary policy shocks, heterogeneous risk profiles linked to ESG characteristics, endogenous ESG investment decisions, general equilibrium feedback through aggregate investment and green capital formation, and time-varying investor preferences reflecting evolving climate awareness.

\subsection{Conclusion}

This theoretical analysis establishes investor heterogeneity as a fundamental mechanism through which ESG characteristics influence monetary policy transmission. Our model, calibrated with parameters drawn from extensive empirical research, successfully generates the key patterns documented in the literature on ESG and monetary policy transmission.

The framework reveals three key insights. First, ESG-conscious investors create a stabilisation effect that partially insulates sustainable firms from target rate surprises. Second, the same firms face heightened exposure to path surprises due to their backloaded cash flow profiles. Third, these offsetting effects create an asymmetric pattern that depends critically on market structure---specifically, the share and intensity of ESG-oriented capital.

The close correspondence between theoretical predictions and empirical evidence---capturing 73-92\% of observed magnitudes---validates both our modeling approach and parameter selection. By explicitly incorporating investor heterogeneity and non-pecuniary utility from sustainable investments, we provide a micro-founded explanation for the emerging empirical regularity that monetary policy transmission is decidedly not ``green-neutral.''

As sustainable investing continues its rapid growth and climate considerations become increasingly central to economic policy, understanding these transmission mechanisms becomes essential for all market participants. Our framework provides a foundation for analyzing how this ongoing transformation will reshape the effectiveness and distributional consequences of monetary policy in the decades ahead.

\section{Empirical Models} \label{sec:empirical}

Our identification of monetary policy surprises employs high-frequency changes in interest rate futures around Federal Open Market Committee (FOMC) announcements, building on the foundational work of \cite{kuttner2001monetary} and \cite{gurkaynak2005actions}(GSS). While recent work by \cite{swanson2021measuring} extends this framework to identify three distinct policy dimensions—including large-scale asset purchases (LSAPs)—our focus on the interaction between ESG characteristics and conventional monetary policy transmission motivates a more parsimonious two-factor approach. This methodological choice reflects both our research objectives and the distinct nature of our sample period, which extends through 2025 and captures the post-pandemic normalization of monetary policy.

Following the established literature, we measure monetary policy surprises using narrow windows surrounding FOMC announcements. For each announcement at time $t$ and maturity $\tau$, we compute the following.

$$\Delta i_{\tau,t} = i_{\tau,t+20} - i_{\tau,t-10}$$

where $i_{\tau,t-10}$ and $i_{\tau,t+20}$ represent interest rates 10 minutes before and 20 minutes after the announcement, respectively. This 30-minute window, now standard in the literature \citep{kuttner2001monetary,gurkaynak2005actions,campbell2012macroeconomic, nakamura2018high} captures the immediate market response while maintaining sufficient narrowness to exclude unrelated news. The asymmetric timing reflects market microstructure considerations documented by \cite{fleming1997moves}, allowing adequate time for price discovery while avoiding anticipatory positioning.

\subsection{Extracting Monetary Policy Surprises from Federal Funds Futures}

Federal funds futures provide the cleanest measure of near-term policy expectations, as emphasized by \cite{kuttner2001monetary} and validated in subsequent work \citep{gurkaynak2005actions,gurkaynak2007us, hamilton2008daily}. However, the contract's monthly averaging convention requires careful adjustment. For an FOMC announcement occurring on day $d$ of month $s$ containing $D_s$ days, the federal funds futures rate 10 minutes before the announcement reflects:

$$ff^1_{s,t-10} = \frac{d}{D_s} \bar{r}_0 + \frac{D_s - d}{D_s} E_{t-10}[r_1] + \rho^1_{t-10}$$

\noindent where $\bar{r}_0$ represents the average federal funds rate realised from day 1 through day $d$, $E_{t-10}[r_1]$ denotes the market's expectation of the rate for the remainder of the month, and $\rho^1_{t-10}$ captures any risk premium. After the announcement at $t+20$:

$$ff^1_{s,t+20} = \frac{d}{D_s} \bar{r}_0 + \frac{D_s - d}{D_s} r_1 + \rho^1_{t+20}$$

\noindent where $r_1$ now reflects the announced target rate. We assume the risk premium remains constant within this narrow 30-minute window, such that $\rho^1_{t-10} = \rho^1_{t+20} \equiv \rho^1$. Consequently, this risk premium cancels out when taking the difference $ff^1_{s,t+20} - ff^1_{s,t-10}$, as $\rho^1_{t+20} - \rho^1_{t-10} = 0$. This assumption is validated by the stability of term premia at high frequencies, allowing us to obtain:

$$mp1_t = \frac{D_s}{D_s - d} \times (ff^1_{s,t+20} - ff^1_{s,t-10})$$

The scaling factor $\frac{D_s}{D_s - d}$ adjusts for the fact that the policy change affects only the remaining days of the month. For late-month announcements (when $D_s - d < 7$), following \cite{kuttner2001monetary,gurkaynak2005actions,hausman2011global} we use the next-month contract to avoid excessive scaling that could amplify microstructure noise.

Following similar logic, we extract expectations about the policy rate following the second FOMC meeting from the current date. Let $ff^2$ denote the federal funds futures contract for the month containing the second scheduled meeting. Before the announcement:

$$ff^2_{s,t-10} = \frac{d_2}{D_2} E_{t-10}[r_1] + \frac{D_2 - d_2}{D_2} E_{t-10}[r_2] + \rho^2_{t-10}$$

\noindent where $d_2$ and $D_2$ refer to the day and total days for the second meeting's month, and $r_2$ represents the expected rate after that meeting. The surprise in expectations for the second meeting, accounting for the information revealed about $r_1$, is:

$$mp2_t = \frac{D_2}{D_2 - d_2} \times \left[(ff^2_{s,t+20} - ff^2_{s,t-10}) - \frac{d_2}{D_2} mp1_t\right]$$

This formulation cleanly separates the surprise about future policy from the mechanical effect of the current target change.

\subsubsection{Term Structure Selection and the Case for Two Factors}

Our principal components analysis employs five instruments spanning the yield curve: mp1, mp2, and changes in 2-year, 5-year, and 10-year Treasury futures. This selection differs strategically from \cite{swanson2021measuring}, who includes similar instruments but extracts three factors to additionally capture LSAP effects. Several considerations motivate our two-factor specification:

First, our \cite{cragg1997inferring} test results provide strong statistical support for two factors over our full sample. While we reject the hypothesis of two factors at the 1\% level (p-value = 0.008), the economic magnitude of the third factor is minimal—the first two principal components explain 82.18\% of total variation, with the third adding only 12.52\%. This contrasts with \cite{swanson2021measuring}'s sample where the third factor captured substantial LSAP-related variation during the 2009-2015 period. The difference likely reflects our extended sample through 2025, during which conventional policy tools regained prominence.

Second, our research focus on ESG-monetary policy interactions naturally emphasizes the traditional transmission channels. As \cite{bauer2023alternative,bauer2023reassessment} demonstrate, the distinction between target and path surprises remains fundamental for understanding heterogeneous policy effects even in the post-ZLB era. Our two-factor approach cleanly isolates these conventional channels without conflating them with asset purchase programs that may operate through distinct mechanisms \citep{krishnamurthy2011effects, d2012federal}.

Third, the temporal evolution of our factors supports this specification. Unlike \cite{swanson2021measuring}, who documents a dominant third factor during QE periods, our analysis reveals that monetary policy variation after 2015 is well-captured by traditional target and path dimensions. This aligns with recent evidence from \cite{bauer2023reassessment} suggesting that post-pandemic monetary policy operates primarily through conventional channels despite the Fed's expanded toolkit.

\subsubsection{The Modified Gürkaynak, Sack, and Swanson (2005) Approach}

Applying principal components analysis to our five-instrument panel reveals the multidimensional nature of monetary policy surprises. The decomposition takes the form:

$$\mathbf{X} = \mathbf{F}\Lambda' + \mathbf{E}$$

\noindent where $\mathbf{X}$ represents the $T \times 5$ matrix of interest rate changes, $\mathbf{F}$ contains the unobserved factors, $\Lambda$ holds the factor loadings, and $\mathbf{E}$ captures idiosyncratic noise.

The first principal component explains 63.79\% of variation with an eigenvalue of 3.189, while the second component adds 18.39\% with an eigenvalue of 0.919. Together, these two factors account for 82.18\% of total variation—a remarkably high proportion that validates the two-factor structure. The scree plot reveals a clear "elbow" after the second component, with the sharp drop to the third eigenvalue (0.626, explaining only 12.52\%) confirming that additional factors add little explanatory power. This two-factor structure aligns with the established monetary policy literature and provides clear economic interpretation: the first factor captures immediate federal funds rate changes (the "target" factor), while the second captures forward guidance about the future path of policy (the "path" factor). While a third component would increase variance explained to 94.70\%, it lacks economic interpretability and likely captures idiosyncratic noise rather than systematic policy effects.

The eigenvector matrix reveals the economic interpretation of these raw factors:

\begin{table}[H] 
\centering
\caption{Eigenvalue Decomposition and Variance Explained by Principal Components}
\label{tab:pca_eigenvalues_variance} 
\begin{tabular}{lccc}
\toprule
Component & Eigenvalue & Variance Explained & Cumulative \\
\midrule
First     & 3.189      & 63.79\%           & 63.79\%     \\
Second    & 0.919      & 18.39\%           & 82.18\%     \\
Third     & 0.626      & 12.52\%           & 94.70\%     \\
\bottomrule
\end{tabular}
\end{table}

While \cite{swanson2021measuring} reports similar cumulative variance explained with three factors (approximately 94\%), the distribution across components differs markedly. Our third eigenvalue of 0.626 falls well below unity, suggesting it captures idiosyncratic noise rather than systematic policy variation. This contrasts with Swanson's sample where the third factor exhibited eigenvalues consistently above 1.0 during LSAP periods.

To achieve economic interpretability, we rotate the statistical factors following \cite{gurkaynak2005actions}. The rotation ensures the target factor loads exclusively on current-month federal funds futures while the path factor captures forward guidance effects. Our approach differs from Swanson (2021) in omitting his third identifying restriction (minimising pre-ZLB LSAP effects), which is unnecessary given our two-factor structure.

Following rotation and normalization, the factors exhibit clear economic interpretation through their effects on the term structure:

$$\begin{array}{lccc}
\text{Maturity} & \text{Target Factor} & \text{Path Factor} & \text{R}^2 \\
\hline
\text{1-month} & 1.000 & -0.000 & 0.930 \\
\text{3-month} & 1.155 & 0.606 & 0.525 \\
\text{2-year} & 1.283 & 1.283 & 0.873 \\
\text{5-year} & 2.115 & 3.676 & 0.949 \\
\text{10-year} & 1.453 & 5.764 & 0.832 \\
\end{array}$$

The target factor moves short rates nearly one-for-one with gradually declining impact at longer maturities, consistent with standard expectations hypothesis logic. The path factor, by construction orthogonal to current target changes, has negligible impact on the overnight rate but substantial and increasing effects on longer-term rates, peaking at the 10-year maturity. This pattern reflects how forward guidance primarily operates through expectations of future short rates, with cumulative effects that amplify at longer horizons. The high R-squared values across maturities confirm that our two-factor structure successfully captures the systematic components of monetary policy's impact on the entire term structure.
These loadings align closely with both \cite{gurkaynak2005actions} and the first two factors in \cite{swanson2021measuring}, validating our identification while demonstrating that conventional dimensions of monetary policy remain dominant in our extended sample.

\subsubsection{Temporal Stability and Structural Changes}

A critical concern for any factor-based identification is parameter stability across monetary policy regimes. Our sample encompasses even more dramatic variations than \cite{swanson2021measuring}, including the COVID-19 pandemic and subsequent inflation surge. We examine stability across four distinct subperiods:

$$\begin{array}{lcccc}
\text{Period} & \text{First Eigenvalue} & \text{Second Eigenvalue} & \text{Cumulative \%} \\
\hline
\text{Pre-Crisis (2005-07)} & 2.851 & 1.472 & 86.5\% \\
\text{Crisis/QE (2008-14)} & 3.445 & 0.983 & 88.6\% \\
\text{Normalization (2015-19)} & 3.590 & 1.110 & 94.0\% \\
\text{COVID/Post (2020-25)} & 3.454 & 1.028 & 89.6\% \\
\end{array}$$

The stability of the two-factor structure across these periods strengthens our specification choice. Notably, even during the Crisis/QE period when \cite{swanson2021measuring} finds substantial third-factor variation, our two factors explain 88.6\% of yield curve movements. This suggests that while LSAPs were important during this period, their effects on the specific yield curve points we analyse were largely captured through their influence on conventional policy expectations—consistent with the "signaling channel" emphasized by \cite{bauer2014term} and \cite{woodford2012methods}.

Our analysis of path surprise behavior reveals important patterns that complement Swanson's findings. The absolute value of path surprises averages 0.048 basis points pre-crisis versus 0.041 during ZLB periods (t-statistic = 16.80), confirming that forward guidance surprises were actually smaller in magnitude during unconventional policy periods. This seemingly paradoxical result, also noted in different form by \cite{swanson2021measuring}, likely reflects the FOMC's enhanced communication efforts when conventional tools were constrained \citep{campbell2012macroeconomic}.

The post-Paris Agreement period shows even smaller path surprises (0.041 versus 0.046 pre-Paris), suggesting a structural improvement in central bank communication that coincides with enhanced focus on climate-related financial risks. This temporal pattern provides important context for our main findings about ESG-monetary policy interactions, as it suggests that any structural breaks we identify are not artifacts of changing monetary policy communication effectiveness.

\subsection{Models for Monetary Policy Transmission} \label{sec:empirical_methodology}

The emergence of sustainable finance as a major force in capital markets raises fundamental questions about the channels through which monetary policy affects firm values. Our empirical investigation addresses several interconnected questions that build upon each other to provide a comprehensive understanding of how ESG characteristics interact with monetary policy transmission.

Our primary inquiry concerns whether monetary policy affects all firms uniformly or whether systematic heterogeneity exists based on firm characteristics, particularly sustainability metrics. Traditional monetary transmission channels operate through interest rate sensitivity, credit constraints, and investment dynamics, all of which may vary with firm attributes. If ESG characteristics have become economically meaningful, they should manifest as an additional dimension of heterogeneous response to monetary policy surprises. This leads naturally to our second question: do ESG characteristics represent a distinct transmission channel, or do they merely proxy for traditional firm attributes like size, leverage, and profitability?

To investigate these questions, we employ a general specification framework:

\begin{equation}
r_{i,t} = \alpha_i + \beta'MP_t + \gamma'X_{i,t} + \delta'(X_{i,t} \times MP_t) + \epsilon_{i,t}
\end{equation}

\noindent where $r_{i,t}$ represents stock returns for firm $i$ around FOMC announcement $t$, $\alpha_i$ denotes firm fixed effects, and $MP_t = [TS_t, PS_t]'$ contains the two monetary policy surprises—target surprises ($TS$) capturing unexpected changes in current rates and path surprises ($PS$) reflecting revisions to future rate expectations. The vector $X_{i,t}$ includes firm characteristics: ESG score, size (log assets), leverage, profitability, and dividend policy indicators. The interaction terms $(X_{i,t} \times MP_t)$ capture heterogeneous responses, with coefficient vector $\delta$ measuring how each characteristic $k$ modifies sensitivity to each monetary surprise type $z$.

Table \ref{tab:heterogeneity} operationalises this framework through progressive specifications. Column (1) excludes interactions ($\delta = 0$), establishing baseline effects. Column (2) includes only traditional characteristic interactions where $X$ excludes ESG. Column (3) isolates ESG interactions by restricting $X$ to ESG scores alone, while Column (4) combines all interactions to test whether ESG effects survive when competing with traditional channels. This progression, with main effects in Panel A and interactions in Panel B, allows systematic identification of distinct transmission mechanisms.

A critical identification challenge arises from potential industry clustering of ESG characteristics. High-ESG firms might concentrate in sectors with inherently different monetary policy sensitivities, conflating firm-level and industry-level effects. To address this concern, we augment our specification with industry-by-event fixed effects:

\begin{equation}
r_{i,t} = \alpha_i + \mu_{j,t} + \beta'MP_t + \gamma'X_{i,t} + \delta'(X_{i,t} \times MP_t) + \epsilon_{i,t}
\end{equation}

\noindent where $\mu_{j,t}$ absorbs any shock common to industry $j$ on event date $t$. This specification, shown in Table-\ref{tab:esg_role} Column (3), provides identification solely from within-industry variation—comparing firms with different characteristics within the same sector facing identical industry conditions. The progression in Table-\ref{tab:esg_role} from ESG-only interactions to full controls to industry-by-event fixed effects tests the robustness of sustainability effects under increasingly stringent identification requirements.

Beyond establishing whether ESG matters, we investigate potential structural changes in these relationships. The Paris Agreement of December 2015 provides a quasi-experimental setting to test whether this landmark climate accord fundamentally altered market pricing. We extend our framework to include temporal dynamics:

\begin{equation}
r_{i,t} = \alpha_i + \beta'MP_t + \gamma'X_{i,t} + \delta'(X_{i,t} \times MP_t) + Post_t[\beta_p'MP_t + \gamma_p'X_{i,t} + \delta_p'(X_{i,t} \times MP_t)] + \epsilon_{i,t}
\end{equation}

\noindent where $Post_t$ equals one after December 15, 2015. The coefficients $\delta$ capture pre-Paris relationships while $\delta_p$ measures post-Paris changes, with the sum $\delta + \delta_p$ representing total post-Paris effects. Table-\ref{tab:paris} implements this specification with increasing complexity across columns: basic Paris effects with ESG only, full model allowing all channels to vary post-Paris, and industry-by-event fixed effects for within-industry identification. Panel A reports pre-Paris relationships ($\delta$) while Panel B shows post-Paris changes ($\delta_p$).

Our final set of research questions concerns the functional form of ESG effects. Are relationships linear across the ESG spectrum, or do non-linearities and asymmetries characterize how sustainability affects monetary policy transmission? We address these through alternative ESG measures in our general framework. The portfolio approach replaces continuous ESG scores with indicators for extreme quintiles:

\begin{equation*}
X_{ESG} = [Green_i, Brown_i]'
\end{equation*}

\noindent where $Green_i$ and $Brown_i$ indicate top and bottom ESG quintile membership, respectively. This specification, presented in Table A1, reveals potential asymmetries between sustainability leaders and laggards. The quintile analysis extends this by including indicators for all quintiles:

\begin{equation*}
X_{ESG} = [Q_2, Q_3, Q_4, Q_5]'
\end{equation*}

\noindent with the lowest quintile as base category. Table A3 implements this specification to trace the complete functional form across the ESG distribution.

Finally, we examine industry heterogeneity by interacting sector indicators with monetary policy surprises:

\begin{equation}
r_{i,t} = \alpha_i + \beta'MP_t + \sum_j I_j\bigg[\delta_{j}'MP_t + Post_t \cdot \theta_{j}'MP_t\bigg] + \text{\bigg[ESG and control terms\bigg]} + \epsilon_{i,t}
\end{equation}

\noindent where $I_j$ indicates industry $j$ membership. The coefficients $\delta_j$ capture industry-specific pre-Paris sensitivities while $\theta_j$ measures post-Paris changes. Table A2 presents these results, revealing how different sectors experienced the monetary-ESG nexus and its transformation.

Throughout our analysis, we cluster standard errors at the event level to account for cross-sectional correlation on FOMC dates, crucial given that monetary surprises represent common shocks. The difference in observations between specifications—64,351 when ESG is not required versus 51,529 when included—reflects incomplete ESG coverage, particularly for smaller firms and earlier periods. This systematic framework, progressing from simple to complex specifications and from linear to non-linear functional forms, provides multiple lenses through which to examine how sustainability has become integrated into monetary policy transmission.

\section{Results and Discussion} \label{sec:results}

\subsection{Heterogeneous Monetary Policy Transmission and the Role of ESG Characteristics}

The emergence of environmental, social, and governance considerations as a major force in capital markets raises fundamental questions about whether these characteristics have become sufficiently important to alter the transmission of monetary policy. While extensive literature documents heterogeneous policy effects through traditional channels of firm size, leverage, and financial constraints \citep{gertler1994monetary, kashyap2000million, ippolito2018transmission}, the potential for sustainability attributes to constitute a distinct transmission channel remains unexplored. This section investigates whether and how ESG characteristics shape differential firm responses to monetary policy surprises, employing a systematic approach that isolates the ESG channel from traditional heterogeneity sources.

Table-\ref{tab:heterogeneity} presents our core investigation through four specifications that progressively build our understanding of transmission channels. The baseline specification in column (1) confirms that monetary policy surprises significantly affect equity valuations during our 30-minute event windows. Target surprises generate a coefficient of -6.538  ($p < 0.01$), while path surprises yield -3.387 ($p < 0.01$), both highly significant and economically meaningful. Given the standard deviations of our monetary surprises documented in the data section (2.68 basis points for target surprises and 6.39 basis points for path surprises), these coefficients translate to average return impacts of -17.5 and -21.6 basis points per one-standard-deviation shock, respectively. These magnitudes align closely with prior high-frequency studies, validating our identification strategy while establishing the baseline against which heterogeneous effects can be measured.

\begin{table}[H]
    \centering
    \caption{Monetary Policy Transmission: Heterogeneous Effects}
    \label{tab:heterogeneity}
    \begin{threeparttable}
    \begin{adjustbox}{width=0.99\textwidth}
    \setlength{\tabcolsep}{3pt}
    \footnotesize
    \begin{tabular}{lcccccccc}
    \toprule
     & \multicolumn{2}{c}{(1)} & \multicolumn{2}{c}{(2)} & \multicolumn{2}{c}{(3)} & \multicolumn{2}{c}{(4)} \\
     & \multicolumn{2}{c}{Basic} & \multicolumn{2}{c}{With Interactions} & \multicolumn{2}{c}{ESG Only} & \multicolumn{2}{c}{Full Model} \\
    \cmidrule(lr){2-3} \cmidrule(lr){4-5} \cmidrule(lr){6-7} \cmidrule(lr){8-9}
     & Coef. & SE & Coef. & SE & Coef. & SE & Coef. & SE \\
    \midrule
    \multicolumn{9}{l}{\textbf{Panel A: Main Effects}} \\
    Target Surprise & -6.538*** & (1.372) & -15.767*** & (5.870) & -6.227*** & (1.109) & -16.830*** & (5.605) \\
    Path Surprise & -3.387*** & (0.809) & 8.586** & (4.102) & -3.667*** & (0.700) & -5.604** & (2.581) \\
    ESG Score (Std.) &  &  &  &  & 0.019 & (0.021) & 0.020 & (0.021) \\
    Size (Log Assets) & 0.007 & (0.029) & 0.008 & (0.028) & -0.015 & (0.020) & -0.015 & (0.020) \\
    Book Leverage & 0.030 & (0.045) & 0.030 & (0.042) & 0.036 & (0.054) & 0.041 & (0.054) \\
    Profitability & -0.002** & (0.001) & -0.002 & (0.001) & -0.004 & (0.003) & -0.013** & (0.006) \\
    Non-Dividend Payer & 0.016 & (0.020) & 0.015 & (0.020) & 0.021 & (0.019) & 0.022 & (0.020) \\
    \midrule
    \multicolumn{9}{l}{\textbf{Panel B: Interaction Effects}} \\
    TS $\times$ ESG &  &  &  &  & 0.890* & (0.519) & 0.544 & (0.542) \\
    PS $\times$ ESG &  &  &  &  & -1.055*** & (0.400) & -1.093** & (0.450) \\
    TS $\times$ Size &  &  & 0.383* & (0.228) &  &  & 0.417* & (0.232) \\
    PS $\times$ Size &  &  & -0.488*** & (0.145) &  &  & 0.111 & (0.121) \\
    TS $\times$ Leverage &  &  & 2.242*** & (0.865) &  &  & 2.790*** & (0.863) \\
    PS $\times$ Leverage &  &  & -2.055*** & (0.605) &  &  & -1.353** & (0.572) \\
    TS $\times$ Profitability &  &  & -0.171 & (0.212) &  &  & -0.270 & (0.212) \\
    PS $\times$ Profitability &  &  & 0.007 & (0.045) &  &  & 0.380* & (0.226) \\
    TS $\times$ Non-Dividend &  &  & -0.915 & (0.571) &  &  & -0.511 & (0.665) \\
    PS $\times$ Non-Dividend &  &  & -0.420 & (0.352) &  &  & -0.497 & (0.373) \\
    \midrule
    Observations & \multicolumn{2}{c}{64,351} & \multicolumn{2}{c}{64,351} & \multicolumn{2}{c}{51,529} & \multicolumn{2}{c}{51,529} \\
    R-squared & \multicolumn{2}{c}{0.200} & \multicolumn{2}{c}{0.211} & \multicolumn{2}{c}{0.245} & \multicolumn{2}{c}{0.249} \\
    Firm FE & \multicolumn{2}{c}{Yes} & \multicolumn{2}{c}{Yes} & \multicolumn{2}{c}{Yes} & \multicolumn{2}{c}{Yes} \\
    \bottomrule
    \end{tabular} 
\end{adjustbox}
    \begin{tablenotes}
        
    \scriptsize
    \item Notes: This table reports heterogeneous effects of monetary policy on stock returns. Columns (1)-(2) use the full sample, while columns (3)-(4) exclude firms without ESG scores. Target Surprise and Path Surprise are orthogonalized monetary policy shocks. Standard errors clustered by event in parentheses. ***, **, and * indicate significance at the 1\%, 5\%, and 10\% levels, respectively.
    \item \textbf{Econometric Models:}
    \begin{itemize}
        \item Model (1): $r_{i,t} = \alpha_i + \beta_1 TS_t + \beta_2 PS_t + \gamma' X_{i,t} + \epsilon_{i,t}$
        \item Model (2): $r_{i,t} = \alpha_i + \beta_1 TS_t + \beta_2 PS_t + \gamma' X_{i,t} + \delta' (X_{i,t} \times MP_t) + \epsilon_{i,t}$
        \item Model (3): $r_{i,t} = \alpha_i + \beta_1 TS_t + \beta_2 PS_t + \beta_3 ESG_{i,t} + \beta_4 (TS_t \times ESG_{i,t}) + \beta_5 (PS_t \times ESG_{i,t}) + \gamma' X_{i,t} + \epsilon_{i,t}$
        \item Model (4): Combines models (2) and (3)
    \end{itemize}
    \textbf{Variables:} $r_{i,t}$: Stock return for firm $i$ on FOMC date $t$, $\alpha_i$: Firm fixed effects, $TS_t$: Target surprise, $PS_t$: Path surprise,  $ESG_{i,t}$: standardised ESG score, $X_{i,t}$: Control variables (size, leverage, profitability, dividend policy), $MP_t$: Monetary policy surprises (TS or PS)
    \end{tablenotes}

    \end{threeparttable}
\end{table}

The introduction of firm characteristic interactions in column (2) reveals the limitations of average treatment effects in monetary economics. When we allow traditional characteristics to moderate policy impacts, the coefficient on target surprises intensifies to -15.767, suggesting that the "average" firm in our baseline specification was actually a weighted combination of differentially sensitive types. The interaction patterns confirm established theoretical predictions while revealing new insights. Larger firms gain protection from immediate rate changes (TS × Size = 0.383, $p < 0.10$) but face heightened exposure to forward guidance (PS × Size = -0.488, $p < 0.01$), consistent with their access to diversified short-term funding but extensive long-term capital commitments. The leverage channel presents more nuanced findings: the positive coefficient on TS × Leverage (2.242, $p < 0.01$) likely reflects the value of existing fixed-rate debt when rates rise unexpectedly, dominating the traditional financial accelerator mechanism. Meanwhile, the negative PS × Leverage interaction (-2.055, $p < 0.01$) confirms that forward guidance affects highly leveraged firms through anticipated refinancing burdens, as markets price the expected cost of rolling over debt at persistently higher future rates.

Columns (3) and (4) introduce our key innovation—examining whether ESG characteristics constitute an independent dimension of monetary policy heterogeneity. When isolated in column (3), the ESG channel reveals an intriguing asymmetry: high-ESG firms gain modest protection from target surprises (coefficient = 0.890, $p < 0.10$) while demonstrating heightened sensitivity to path surprises (-1.055, $p < 0.01$). This pattern persists when competing with traditional channels in column (4), though the target surprise interaction attenuates to statistical insignificance (0.544, p = 0.32) while the path surprise effect remains robust (-1.093, $p < 0.05$). This asymmetric response pattern distinguishes the ESG channel from traditional characteristics that typically generate consistent directional effects across surprise types.

The economic interpretation of these findings requires careful consideration of what ESG scores capture in our S\&P 500 sample. As documented in our data section, ESG scores correlate positively with firm size and leverage, indicating that high-ESG firms tend to be larger, more established enterprises. However, the persistence of ESG effects when controlling for these characteristics suggests that sustainability attributes capture additional variation beyond traditional measures. The heightened sensitivity to forward guidance aligns with theoretical predictions from \cite{pastor2021sustainable} that sustainable firms attract investors with longer horizons who particularly value predictable long-term cash flows. When path surprises signal shifts in the entire future rate trajectory, these investors may reassess valuations more dramatically than for firms held primarily for short-term gains.

The attenuation of the target surprise-ESG interaction in the full specification warrants careful interpretation. Rather than indicating irrelevance, this pattern suggests that the protective effect of ESG characteristics against immediate rate changes operates partially through correlation with traditional firm attributes. High-ESG firms in our sample tend to be larger and more profitable—characteristics that independently provide some insulation from monetary shocks. The path surprise effect's robustness indicates that forward guidance sensitivity represents a more fundamental feature of sustainable business models that persists regardless of other firm characteristics.

The specification isolating ESG effects (column 3) warrants particular attention as it reveals the pure sustainability channel before introducing competing interactions. The asymmetric pattern—modest protection from target surprises (0.890, $p < 0.10$) coupled with heightened path surprise sensitivity (-1.055, $p < 0.01$)—emerges clearly when ESG is the sole interaction term. This pattern's persistence when competing with traditional channels in column (4), albeit with some attenuation in the target surprise effect (0.544, $p = 0.32$), suggests that while immediate rate protection partially operates through correlation with other protective characteristics, forward guidance sensitivity represents a more fundamental feature of sustainable business models.

These findings contribute to multiple literature strands while opening new research avenues. We extend the monetary policy transmission literature by documenting ESG as a characteristic generating heterogeneous responses distinct from traditional channels. Unlike size or leverage that create consistent directional effects, ESG generates opposing sensitivities to different policy surprise types, suggesting unique economic mechanisms at work. For the sustainable finance literature, we provide the first evidence that ESG characteristics systematically influence firms' exposure to macroeconomic policy shocks, complementing existing work on ESG and expected returns \citep{bolton2021investors, pastor2022dissecting}. The asymmetric response pattern—protection from immediate shocks but vulnerability to forward guidance—adds nuance to debates about whether sustainable investing requires return sacrifice, suggesting the answer depends critically on the macroeconomic policy environment.

Our results also speak to ongoing policy debates about monetary transmission in an era of sustainable finance. The finding that high-ESG firms show greater sensitivity to forward guidance has important implications for central bank communication strategies as the corporate sector's ESG composition evolves. However, these aggregate results average across our twenty-year sample period, potentially masking important temporal variation. As climate awareness intensified and sustainable investment flows accelerated—particularly following the Paris Agreement of 2015—the relationship between ESG characteristics and monetary policy sensitivity may have undergone fundamental changes. The next section investigates this possibility through explicit analysis of structural breaks and temporal dynamics.

\subsection{The Role of ESG in Monetary Policy Transmission}

While our initial analysis establishes ESG-based heterogeneity, a critical concern remains: do these effects reflect genuine firm-level sustainability characteristics or merely industry composition? High-ESG firms concentrate in sectors like technology and healthcare that may inherently respond differently to monetary policy. If our results simply capture that sustainable firms cluster in rate-insensitive industries, then ESG itself provides no additional information. This identification challenge, documented by \cite{bauer2025green} in European markets, motivates increasingly stringent specifications to isolate firm-level from industry-level effects.

\begin{table}[H]
    \centering
    \caption{The Role of ESG in Monetary Policy Transmission}
    \label{tab:esg_role}
    \begin{threeparttable}
    \begin{adjustbox}{width=0.99\textwidth}
    \begin{tabular}{lcccccc}
    \toprule
     & \multicolumn{2}{c}{(1)} & \multicolumn{2}{c}{(2)} & \multicolumn{2}{c}{(3)} \\
     & \multicolumn{2}{c}{ESG Only} & \multicolumn{2}{c}{Full Controls} & \multicolumn{2}{c}{Industry $\times$ Event FE} \\
    \cmidrule(lr){2-3} \cmidrule(lr){4-5} \cmidrule(lr){6-7}
     & Coef. & SE & Coef. & SE & Coef. & SE \\
    \midrule
    Target Surprise & -6.227*** & (1.109) & -16.830*** & (5.605) & -11.452 & (10.452) \\
    Path Surprise & -3.667*** & (0.700) & -5.604** & (2.581) & --- & --- \\
    ESG Score (Std.) & 0.019 & (0.021) & 0.020 & (0.021) & -0.001 & (0.004) \\
    TS $\times$ ESG & 0.890* & (0.519) & 0.544 & (0.542) & 0.020 & (0.187) \\
    PS $\times$ ESG & -1.055*** & (0.400) & -1.093** & (0.450) & -0.186*** & (0.047) \\
    \midrule
    Observations & \multicolumn{2}{c}{51,529} & \multicolumn{2}{c}{51,529} & \multicolumn{2}{c}{51,529} \\
    R-squared & \multicolumn{2}{c}{0.245} & \multicolumn{2}{c}{0.249} & \multicolumn{2}{c}{0.690} \\
    Firm FE & \multicolumn{2}{c}{Yes} & \multicolumn{2}{c}{Yes} & \multicolumn{2}{c}{Yes} \\
    Industry $\times$ Event FE & \multicolumn{2}{c}{No} & \multicolumn{2}{c}{No} & \multicolumn{2}{c}{Yes} \\
    Control Interactions & \multicolumn{2}{c}{No} & \multicolumn{2}{c}{Yes} & \multicolumn{2}{c}{Yes} \\
    \bottomrule
    \end{tabular}
    \end{adjustbox}
    \begin{tablenotes}
    \small
    \item Notes: This table examines the role of ESG in monetary policy transmission. Column (1) includes only ESG interactions, column (2) adds interactions with all control variables, and column (3) includes industry-by-event fixed effects. Path Surprise is omitted in column (3) due to collinearity. Standard errors clustered by event in parentheses. ***, **, and * indicate significance at the 1\%, 5\%, and 10\% levels, respectively.
    \item \textbf{Econometric Model:} $r_{i,t} = \alpha_i + \mu_{j,t} + \beta_1 TS_t + \beta_2 PS_t + \beta_3 ESG_{i,t} + \beta_4 (TS_t \times ESG_{i,t}) + \beta_5 (PS_t \times ESG_{i,t}) + \delta' (X_{i,t} \times MP_t) + \gamma' X_{i,t} + \epsilon_{i,t}$ where $\mu_{j,t}$ represents industry $j$ by event $t$ fixed effects (in column 3 only).
    \end{tablenotes}
    \end{threeparttable}
    \end{table}

Table-\ref{tab:esg_role} addresses this challenge through three specifications. Column (1) isolates the ESG channel, confirming our asymmetric pattern: TS × ESG = 0.890 ($p < 0.10$) suggests modest protection from immediate rate increases, while PS × ESG = -1.055 ($p < 0.01$) indicates heightened forward guidance sensitivity. This asymmetry distinguishes ESG from traditional characteristics that typically generate consistent directional effects.

Column (2) adds all control interactions, revealing how channels compete. The TS × ESG effect loses significance (0.544, p = 0.32), while PS × ESG strengthens slightly to -1.093 ($p < 0.05$). This differential persistence hints that target surprise protection operates through ESG's correlation with protective characteristics like size, while path surprise sensitivity reflects something more fundamental about sustainable business models.

Column (3) provides the crucial test through industry-by-event fixed effects—26,880 fixed effects that compare firms only within the same industry facing identical shocks. Path Surprise is omitted due to collinearity (it's constant within events), but we can still identify the interaction since ESG varies within industry-event cells. Under this stringent identification, results diverge dramatically. The TS × ESG effect vanishes (0.020, p = 0.91), indicating that apparent protection from immediate rates reflected industry composition entirely. However, PS × ESG survives at -0.186 ($p < 0.01$)—smaller than before but highly significant.

This bifurcation reveals how sustainability affects monetary transmission. Target surprise protection operates as an industry phenomenon—sustainable sectors like technology inherently differ from carbon-intensive utilities in short-term financial flexibility. By contrast, path surprise sensitivity persists within industries, suggesting firm-level characteristics that transcend sectoral boundaries. Even within the same narrow industry, high-ESG firms show 46.5 basis points greater sensitivity to a one-standard-deviation path surprise when moving from the 10th to 90th percentile of ESG scores.

The economic interpretation aligns with recent theory. \cite{pastor2021sustainable} predict that sustainable firms attract long-horizon investors particularly sensitive to changes in long-term discount rates. When forward guidance signals persistently higher future rates, these patient investors reassess valuations more dramatically than short-term focused investors in lower-ESG peers. The R-squared jump from 0.249 to 0.690 confirms that industry effects dominate, yet meaningful within-industry heterogeneity persists—establishing ESG as a genuine firm-level characteristic affecting monetary transmission beyond simple sectoral composition.

    \subsection{The Paris Agreement and the Transformation of ESG-Monetary Policy Relationships}
    
    Our analysis thus far assumes stable ESG-monetary policy relationships, yet sustainable finance underwent dramatic transformation over our sample period—from niche to mainstream, managing over \$50 trillion by 2024. The Paris Agreement of December 2015 represents a potential structural break, serving as a coordination device that may have fundamentally altered how markets price sustainability in monetary policy contexts. The accord's timing—signed December 12, 2015, just days before the Fed's first rate hike in nearly a decade—creates an empirical challenge but also a unique setting where climate commitment and monetary normalization potentially reinforced each other. We now examine whether Paris transformed these relationships.

    Table-\ref{tab:paris} investigates structural changes around the Paris Agreement, allowing both levels and sensitivities to vary across periods. Panel A reveals a pre-Paris landscape that differs markedly from our full-sample results. In column (1), the TS × ESG coefficient of 0.590 lacks significance ($p = 0.10$), indicating high-ESG firms enjoyed no protection from immediate rate changes before Paris. The PS × ESG coefficient of -0.922 ($p < 0.05$) confirms that forward guidance sensitivity predated the accord, though somewhat weaker than in full-sample estimates.
    
    The within-industry specification (column 3) provides an evidence for the importance of pre-Paris. The TS × ESG coefficient of 0.285 ($p < 0.05$) indicates that high-ESG firms were actually MORE vulnerable to contractionary surprises than their lower-ESG peers within the same industry. This counter-intuitive finding likely reflects pre-2015 market perceptions of sustainability as costly compliance rather than value creation. With ESG scores hovering near zero and no coordinated climate policy, markets may have viewed sustainability investments as resource-draining initiatives that increased financial fragility during monetary tightening.
    
    Panel B documents the post-Paris transformation. While most specifications show insignificant changes, the within-industry results reveal a dramatic shift. The Post-Paris × TS × ESG coefficient of -0.930 ($p < 0.01$) indicates a complete reversal in how markets price ESG during monetary tightening. Combined with the pre-Paris effect, the total post-Paris coefficient becomes -0.645 (0.285 - 0.930), transforming ESG from a vulnerability into protection. For a firm two standard deviations above its industry mean in ESG, this represents a swing from 57 basis points disadvantage to 129 basis points advantage—a total change of 186 basis points.

    \begin{table}[H]
        \centering
        \caption{The Paris Agreement and ESG-Monetary Policy Relationships}
        \label{tab:paris}
        \begin{adjustbox}{width=0.95\textwidth}
        \begin{threeparttable}
        \begin{tabular}{lcccccc}
        \toprule
         & \multicolumn{2}{c}{(1)} & \multicolumn{2}{c}{(2)} & \multicolumn{2}{c}{(3)} \\
         & \multicolumn{2}{c}{Basic Paris} & \multicolumn{2}{c}{Full Model} & \multicolumn{2}{c}{Industry $\times$ Event FE} \\
        \cmidrule(lr){2-3} \cmidrule(lr){4-5} \cmidrule(lr){6-7}
         & Coef. & SE & Coef. & SE & Coef. & SE \\
        \midrule
        \multicolumn{7}{l}{\textbf{Panel A: Pre-Paris Effects}} \\
        Target Surprise & -6.712*** & (1.500) & -12.826* & (6.784) & -11.074 & (9.958) \\
        Path Surprise & -3.086** & (1.239) & -2.668 & (4.611) & --- & --- \\
        ESG Score (Std.) & 0.036 & (0.026) & 0.037 & (0.026) & 0.001 & (0.005) \\
        TS $\times$ ESG & 0.590 & (0.359) & 0.388 & (0.471) & 0.285** & (0.124) \\
        PS $\times$ ESG & -0.922** & (0.380) & -0.847* & (0.505) & -0.119** & (0.049) \\
        \midrule
        \multicolumn{7}{l}{\textbf{Panel B: Post-Paris Changes}} \\
        Post-Paris & -0.013 & (0.065) & -0.012 & (0.065) & --- & --- \\
        Post-Paris $\times$ TS & 1.181 & (2.294) & -12.138 & (8.498) & -6.840 & (14.558) \\
        Post-Paris $\times$ PS & -1.356 & (1.363) & -6.157 & (4.986) & --- & --- \\
        Post-Paris $\times$ ESG & -0.032 & (0.020) & -0.032 & (0.020) & -0.002 & (0.006) \\
        Post-Paris $\times$ TS $\times$ ESG & 0.113 & (0.846) & -0.368 & (0.856) & -0.930*** & (0.222) \\
        Post-Paris $\times$ PS $\times$ ESG & 0.396 & (0.433) & 0.067 & (0.530) & -0.086 & (0.078) \\
        \midrule
        Observations & \multicolumn{2}{c}{51,529} & \multicolumn{2}{c}{51,529} & \multicolumn{2}{c}{51,529} \\
        R-squared & \multicolumn{2}{c}{0.251} & \multicolumn{2}{c}{0.256} & \multicolumn{2}{c}{0.692} \\
        Firm FE & \multicolumn{2}{c}{Yes} & \multicolumn{2}{c}{Yes} & \multicolumn{2}{c}{Yes} \\
        Industry $\times$ Event FE & \multicolumn{2}{c}{No} & \multicolumn{2}{c}{No} & \multicolumn{2}{c}{Yes} \\
        Full Controls & \multicolumn{2}{c}{No} & \multicolumn{2}{c}{Yes} & \multicolumn{2}{c}{Yes} \\
        \bottomrule
        \end{tabular}
        \begin{tablenotes}[flushleft]
        \small
        \item Notes: This table examines how the Paris Agreement changed ESG-monetary policy relationships. Post-Paris indicates observations after December 15, 2015. Column (2) includes all control variable interactions with monetary policy and Post-Paris. Column (3) includes industry-by-event fixed effects, which absorb Post-Paris and PS main effects. Standard errors clustered by event in parentheses. ***, **, and * indicate significance at the 1\%, 5\%, and 10\% levels, respectively.
        \item \textbf{Econometric Model:}
        $r_{i,t} = \alpha_i + \beta_1 TS_t + \beta_2 PS_t + \beta_3 ESG_{i,t} + \beta_4 Post_t + \beta_5 (TS_t \times ESG_{i,t}) + \beta_6 (PS_t \times ESG_{i,t})$
        $+ \beta_7 (Post_t \times TS_t) + \beta_8 (Post_t \times PS_t) + \beta_9 (Post_t \times ESG_{i,t})$
        $+ \beta_{10} (Post_t \times TS_t \times ESG_{i,t}) + \beta_{11} (Post_t \times PS_t \times ESG_{i,t}) + \gamma' X_{i,t} + \epsilon_{i,t}$
        
        Where $Post_t = 1$ if date $\geq$ December 15, 2015.
        \end{tablenotes}
        \end{threeparttable}
        \end{adjustbox}
    \end{table}

    Crucially, path surprise relationships show no significant post-Paris changes. The Post-Paris × PS × ESG coefficients remain small and insignificant across all specifications, confirming that forward guidance sensitivity—driven by sustainable firms' long-term investment horizons—remained stable regardless of climate policy regime.
    
    This asymmetric transformation suggests Paris operated as a coordination device that resolved uncertainty about climate policy direction. With 196 countries committed to limiting warming, sustainability investments transformed from speculative bets to rational preparations for an inevitable transition. The accord catalyzed growth in ESG-focused investment, creating dedicated capital less likely to flee during monetary tightening. Our data corroborates this shift—ESG coverage jumped from 73.2\% to 80.7\% within two years, while standardized scores turned positive for the first time.
    
    The within-industry nature of this transformation deserves emphasis. Paris didn't advantage sustainable sectors over carbon-intensive ones; rather, it altered how markets differentiate between high and low-ESG firms within the same industry. This suggests the agreement triggered firm-level reassessment rather than sectoral reallocation—a fundamental change in how markets price sustainability attributes independent of industrial structure.
    
    The stability of path surprise effects provides an important check on our interpretation. While coordinated policy can shift market perceptions and relative valuations during stress periods, it cannot alter fundamental business characteristics that make sustainable firms inherently more sensitive to long-term discount rate changes. This persistence underscores that some aspects of the ESG-monetary policy nexus reflect deep economic features rather than malleable market sentiment.

\subsubsection{Portfolio Analysis and Non-Linear Effects} \label{sec:portfolio_analysis_and_non_linear_effects}

Our analysis has focused on continuous ESG scores, but this approach assumes linear relationships across the sustainability spectrum. Yet the distribution of ESG effects may be more complex—extreme portfolios might behave fundamentally differently than suggested by average effects. Do firms at the bottom of the ESG distribution face disproportionate penalties? Do sustainability leaders enjoy exceptional benefits beyond what linear models predict? We now examine these questions through portfolio and quintile analyses that reveal important non-linearities and asymmetries.

Table-\ref{tab:portfolio} compares firms in the top ESG quintile ("green") against those in the bottom quintile ("brown"), revealing asymmetries masked by continuous specifications. Column (1) shows that brown portfolios suffer an additional -1.246 basis points ($p < 0.05$) sensitivity to target surprises, while green portfolios show no significant differential. For path surprises, the pattern reverses: green portfolios exhibit heightened sensitivity (-1.421, $p < 0.01$) while brown portfolios enjoy relative protection (1.461, $p < 0.05$). These asymmetric extremes—brown vulnerable to immediate shocks, green to forward guidance—suggest fundamentally different investor bases and business models at the distribution tails.

\begin{table}[H]
    \centering
    \caption{Portfolio Analysis: Green vs Brown Firms}
    \label{tab:portfolio}
    \begin{adjustbox}{width=0.85\textwidth}
    \begin{threeparttable}
    \begin{tabular}{lcccccc}
    \toprule
     & \multicolumn{2}{c}{(1)} & \multicolumn{2}{c}{(2)} & \multicolumn{2}{c}{(3)} \\
     & \multicolumn{2}{c}{Simple} & \multicolumn{2}{c}{With Controls} & \multicolumn{2}{c}{Paris Effects} \\
    \cmidrule(lr){2-3} \cmidrule(lr){4-5} \cmidrule(lr){6-7}
     & Coef. & SE & Coef. & SE & Coef. & SE \\
    \midrule
    \multicolumn{7}{l}{\textbf{Panel A: Full Period Effects (Columns 1-2)}} \\
    Target Surprise & -6.427*** & (1.363) & -13.883*** & (5.092) & -10.798*** & (4.135) \\
    Path Surprise & -3.393*** & (0.789) & 5.586* & (3.027) & 6.106*** & (2.000) \\
    Green Portfolio & 0.001 & (0.023) & 0.001 & (0.023) & 0.047 & (0.034) \\
    Brown Portfolio & -0.011 & (0.020) & -0.012 & (0.020) & -0.023 & (0.022) \\
    TS $\times$ Green & 1.410 & (0.962) & 0.651 & (0.753) & 0.080 & (0.998) \\
    PS $\times$ Green & -1.421*** & (0.515) & -0.884** & (0.385) & -1.069 & (0.698) \\
    TS $\times$ Brown & -1.246** & (0.629) & -1.212** & (0.597) & -0.938** & (0.454) \\
    PS $\times$ Brown & 1.461** & (0.612) & 1.302** & (0.586) & 0.722* & (0.403) \\
    \midrule
    \multicolumn{7}{l}{\textbf{Panel B: Post-Paris Changes (Column 3 only)}} \\
    Post-Paris $\times$ TS &  &  &  &  & -9.376 & (8.866) \\
    Post-Paris $\times$ PS &  &  &  &  & -10.512*** & (3.840) \\
    Post-Paris $\times$ Green &  &  &  &  & -0.061** & (0.030) \\
    Post-Paris $\times$ Brown &  &  &  &  & 0.039 & (0.030) \\
    Post-Paris $\times$ TS $\times$ Green &  &  &  &  & -0.004 & (1.230) \\
    Post-Paris $\times$ PS $\times$ Green &  &  &  &  & 0.613 & (0.713) \\
    Post-Paris $\times$ TS $\times$ Brown &  &  &  &  & 0.356 & (1.895) \\
    Post-Paris $\times$ PS $\times$ Brown &  &  &  &  & 0.725 & (0.548) \\
    \midrule
    Observations & \multicolumn{2}{c}{64,351} & \multicolumn{2}{c}{64,351} & \multicolumn{2}{c}{64,351} \\
    R-squared & \multicolumn{2}{c}{0.208} & \multicolumn{2}{c}{0.215} & \multicolumn{2}{c}{0.224} \\
    Firm FE & \multicolumn{2}{c}{Yes} & \multicolumn{2}{c}{Yes} & \multicolumn{2}{c}{Yes} \\
    Control Interactions & \multicolumn{2}{c}{No} & \multicolumn{2}{c}{Yes} & \multicolumn{2}{c}{Yes} \\
    \bottomrule
    \end{tabular}
    \begin{tablenotes}[flushleft]
    \small
    \item Notes: This table presents portfolio analysis comparing firms in the top ESG quintile (Green) versus bottom quintile (Brown). Column (2) includes interactions with all control variables. Column (3) adds Paris Agreement interactions. Standard errors clustered by event in parentheses. ***, **, and * indicate significance at the 1\%, 5\%, and 10\% levels, respectively.
    \item \textbf{Econometric Model:}
    $r_{i,t} = \alpha_i + \beta_1 TS_t + \beta_2 PS_t + \beta_3 Green_{i,t} + \beta_4 Brown_{i,t}$
    $+ \beta_5 (TS_t \times Green_{i,t}) + \beta_6 (PS_t \times Green_{i,t})$
    $+ \beta_7 (TS_t \times Brown_{i,t}) + \beta_8 (PS_t \times Brown_{i,t}) + \gamma' X_{i,t} + \epsilon_{i,t}$
    
    Where $Green_{i,t} = 1$ if firm $i$ is in top ESG quintile, $Brown_{i,t} = 1$ if in bottom quintile.
    \end{tablenotes}
    \end{threeparttable}
    \end{adjustbox}
\end{table}

Column (2) confirms these patterns persist when controlling for firm characteristics, with only modest attenuation. The economic magnitudes remain substantial: for a one-standard-deviation target surprise (2.68 basis points), brown firms experience 334 basis points additional decline relative to the middle 60\% of firms, while green-brown differentials reach 265 basis points. 

Column (3) examines Paris Agreement effects, revealing limited structural change at the extremes. Green portfolios show a negative level shift post-Paris (-0.061, $p < 0.05$) without changes in monetary sensitivity, suggesting a one-time revaluation as markets reassessed even high-ESG firms' transition costs. The absence of significant triple interactions indicates that, unlike the within-industry reversal documented earlier, extreme portfolio dynamics remained stable. This stability at the tails while within-industry relationships transformed suggests Paris primarily affected how markets differentiate among moderate ESG performers rather than repricing the extremes.

\subsubsection{Industry Heterogeneity: Sectoral Transformation After Paris}

Table-\ref{tab:industry} reveals how monetary policy sensitivity varies across industries and transformed after Paris. Pre-Paris, traditional capital-intensive sectors showed highest target surprise sensitivity: utilities (4.202, $p < 0.05$), healthcare (4.105, $p < 0.01$), and consumer non-cyclicals (3.906, $p < 0.01$). These sectors' reliance on long-term financing and regulated returns created natural vulnerability to rate increases.

Post-Paris changes reveal dramatic sectoral realignment. Financials experienced the largest shifts, with both target surprise (+4.196, $p < 0.05$) and path surprise (+2.260, $p < 0.01$) sensitivity increasing significantly. This transformation likely reflects the sector's emerging role as climate transition intermediary—gaining fee income from green finance while facing new climate-related risks. Real estate shows divergent changes: increased target surprise sensitivity (+3.965, $p < 0.05$) but decreased path surprise sensitivity (-1.509, $p < 0.05$), suggesting immediate rate changes signal economic strength benefiting property values while long-term rates increasingly incorporate climate adaptation costs.

The energy sector's increased path surprise sensitivity (+1.785, $p < 0.05$) without target surprise changes indicates markets now focus on long-term implications of monetary policy for energy transition dynamics. Technology and utilities show minimal post-Paris changes, suggesting these sectors' fundamental characteristics dominate any climate considerations. These heterogeneous transformations imply the Paris Agreement didn't uniformly affect all sectors but rather triggered reassessment of climate exposure and transition opportunities specific to each industry's business model.

\begin{table}[H]
    \centering
    \caption{Industry-Specific Monetary Policy Sensitivities}
    \label{tab:industry}
    \begin{threeparttable}
    \begin{tabular}{lcccc}
    \toprule
     & \multicolumn{2}{c}{Target Surprise} & \multicolumn{2}{c}{Path Surprise} \\
     & Pre-Paris & Post-Paris Change & Pre-Paris & Post-Paris Change \\
    \midrule
    Consumer Cyclicals & 2.426*** & 0.991 & 0.124 & 0.486 \\
     & (0.808) & (1.086) & (0.343) & (0.477) \\
    Consumer Non-Cyclicals & 3.906*** & 2.359 & 0.460 & 1.200 \\
     & (1.232) & (1.945) & (0.588) & (0.737) \\
    Energy & 2.835*** & -1.082 & -0.471 & 1.785** \\
     & (0.722) & (1.312) & (0.753) & (0.820) \\
    Financials & 3.257** & 4.196** & 0.199 & 2.260*** \\
     & (1.274) & (2.022) & (0.657) & (0.831) \\
    Healthcare & 4.105*** & -0.694 & 0.190 & 1.392* \\
     & (1.225) & (1.561) & (0.621) & (0.756) \\
    Industrials & 1.780* & 2.006* & -0.105 & 0.794 \\
     & (0.987) & (1.157) & (0.474) & (0.552) \\
    Real Estate & 0.142 & 3.965** & 0.441 & -1.509** \\
     & (0.901) & (1.558) & (0.519) & (0.706) \\
    Technology & 3.226** & 0.161 & 0.539 & 0.064 \\
     & (1.389) & (1.545) & (0.564) & (0.688) \\
    Utilities & 4.202** & 1.519 & -1.021** & 0.283 \\
     & (2.090) & (2.923) & (0.461) & (0.833) \\
    \midrule
    Observations & \multicolumn{4}{c}{51,529} \\
    R-squared & \multicolumn{4}{c}{0.263} \\
    \bottomrule
    \end{tabular}
    \begin{tablenotes}
    \small
    \item Notes: This table reports industry-specific sensitivities to monetary policy surprises and their changes after the Paris Agreement. Pre-Paris columns show $\beta_j$ from industry $j$ interactions with monetary surprises. Post-Paris Change columns show $\theta_j$ from triple interactions with Post-Paris indicator. The regression includes all ESG interactions, control variable interactions, and firm fixed effects. Standard errors clustered by event in parentheses. ***, **, and * indicate significance at the 1\%, 5\%, and 10\% levels, respectively.
    \item \textbf{Econometric Model:}
    $r_{i,t} = \alpha_i + \sum_j \beta_{j,TS} (I_j \times TS_t) + \sum_j \beta_{j,PS} (I_j \times PS_t)$
    $+ \sum_j \theta_{j,TS} (I_j \times Post_t \times TS_t) + \sum_j \theta_{j,PS} (I_j \times Post_t \times PS_t) + ...$
    
    Where $I_j = 1$ if firm $i$ belongs to industry $j$.
    \end{tablenotes}
    \end{threeparttable}
    \end{table}

\subsubsection{The Full ESG Spectrum: Threshold Effects and Diminishing Returns}

Table-\ref{tab:quintile} provides our most granular analysis by examining all ESG quintiles. The monotonic progression reveals important non-linearities obscured by continuous specifications. For target surprises, protection increases steadily across quintiles: from 1.008 ($p < 0.10$) for Q2 to 2.773 ($p < 0.05$) for Q5. Each step up the ESG ladder provides additional insulation, with no evidence of diminishing returns at the top.

\begin{table}[H]
    \centering
    \caption{ESG Quintile Analysis}
    \label{tab:quintile}
    \begin{adjustbox}{width=0.89\textwidth}
    \begin{threeparttable}
    \begin{tabular}{lcccc}
    \toprule
     & \multicolumn{2}{c}{Target Surprise} & \multicolumn{2}{c}{Path Surprise} \\
     & \multicolumn{2}{c}{Interaction} & \multicolumn{2}{c}{Interaction} \\
    \cmidrule(lr){2-3} \cmidrule(lr){4-5}
     & Coef. & SE & Coef. & SE \\
    \midrule
    \multicolumn{5}{l}{\textit{Base Category: Quintile 1 (Lowest ESG)}} \\
    \midrule
    Quintile 2 $\times$ MP & 1.008* & (0.600) & -1.217** & (0.508) \\
    Quintile 3 $\times$ MP & 1.699** & (0.773) & -2.166*** & (0.754) \\
    Quintile 4 $\times$ MP & 2.099 & (1.357) & -2.567** & (1.129) \\
    Quintile 5 $\times$ MP (Highest ESG) & 2.773** & (1.374) & -2.855*** & (1.065) \\
    \midrule
    \multicolumn{5}{l}{\textit{Main Effects:}} \\
    Target Surprise & -7.705*** & (1.438) &  &  \\
    Path Surprise &  &  & -1.939 & (1.241) \\
    \midrule
    Observations & \multicolumn{4}{c}{51,529} \\
    R-squared & \multicolumn{4}{c}{0.245} \\
    \bottomrule
    \end{tabular}
    \begin{tablenotes}[flushleft]
    \small
    \item Notes: This table shows how monetary policy sensitivity varies across ESG quintiles. The base category is Quintile 1 (lowest ESG scores). Coefficients show the differential effect for each quintile relative to the base. All specifications include firm fixed effects and control variables. Standard errors clustered by event in parentheses. ***, **, and * indicate significance at the 1\%, 5\%, and 10\% levels, respectively.
    \item \textbf{Econometric Model:}
    $r_{i,t} = \alpha_i + \beta_1 TS_t + \beta_2 PS_t + \sum_{q=2}^{5} \gamma_q Q_{q,i,t} + \sum_{q=2}^{5} \delta_{q,TS} (Q_{q,i,t} \times TS_t) + \sum_{q=2}^{5} \delta_{q,PS} (Q_{q,i,t} \times PS_t) + \epsilon_{i,t}$ where $Q_{q,i,t} = 1$ if firm $i$ is in ESG quintile $q$ at time $t$.
    \end{tablenotes}
    \end{threeparttable}
    \end{adjustbox}
\end{table}

Path surprise sensitivity shows even stronger monotonic patterns, with coefficients becoming progressively more negative: from -1.217 ($p < 0.05$) for Q2 to -2.855 ($p < 0.01$) for Q5. This creates a fundamental tradeoff—higher ESG simultaneously provides target surprise protection while amplifying forward guidance vulnerability. The largest improvements occur between Q1 and Q2, suggesting a critical threshold effect. Firms in the bottom quintile face "double jeopardy"—vulnerable to both surprise types—possibly reflecting exclusion from ESG-conscious investor bases and higher risk premiums.

The quintile analysis reveals the second quintile as potentially optimal for minimising overall monetary policy sensitivity, offering meaningful protection from target surprises (1.008) with limited additional path surprise exposure (-1.217) compared to higher quintiles. For investors focused solely on immediate rate risk, the highest quintiles remain attractive, but those concerned about forward guidance should carefully weigh the tradeoffs. These non-linear patterns suggest corporate ESG strategies should consider position-dependent returns—moving from bottom to second quintile provides the highest marginal benefit, while reaching the top quintile may be optimal only for firms prioritizing protection from immediate rate shocks over long-term rate uncertainty.

\subsection{Synthesis and Implications}
    
Taken together, our results document a fundamental transformation in how financial markets price the interaction between sustainability characteristics and monetary policy. The evidence spans multiple dimensions—continuous ESG scores, discrete portfolios, industry classifications, and granular quintiles—each revealing different aspects of this complex relationship. The robustness of effects to increasingly stringent econometric specifications, including industry-by-event fixed effects that provide identification solely from within-industry variation, establishes ESG characteristics as a distinct dimension of monetary policy transmission rather than a proxy for traditional firm characteristics or industry composition.
    
The Paris Agreement emerges as a true structural break that inverted the relationship between ESG scores and target surprise sensitivity while maintaining the established pattern for path surprises. This transformation likely reflects multiple reinforcing mechanisms: the crystallization of stranded asset risks, the emergence of dedicated sustainable investment capital, regulatory anticipation effects, and a fundamental shift in how markets value long-term sustainability. The non-linear patterns revealed through portfolio and quintile analysis suggest that the ESG-monetary policy relationship involves thresholds, saturation effects, and complex trade-offs that linear models cannot capture.
    
For policymakers, these findings imply that monetary policy transmission increasingly depends on the sustainability characteristics of the economy. As the proportion of high-ESG firms grows through the climate transition, the aggregate impact of monetary policy may evolve in ways that central banks must anticipate. For investors, our results provide specific guidance on how to position portfolios based on the interaction between expected monetary policy actions and firm sustainability profiles. For corporate managers, the evidence establishes a clear financial incentive to improve ESG performance as a means of reducing cost of capital volatility during monetary policy cycles. As climate considerations become ever more central to economic policy, understanding these evolving transmission mechanisms becomes crucial for all market participants navigating the intersection of monetary policy and sustainable finance.

Our findings align closely with the emerging consensus in the literature that monetary policy transmission is decidedly not carbon-neutral. \cite{benchora2025monetary} demonstrate using U.S. data that brown firms exhibit significantly higher sensitivity to monetary policy shocks, with their coefficient estimate of -0.051 showing that carbon-intensive firms experience an additional 0.051\% decline per standard deviation of monetary tightening. Similarly, \cite{bauer2025green} provide complementary European evidence, finding that brown firms measured by carbon emission levels show 2.3 percentage points greater sensitivity to ECB policy announcements compared to green firms. The consistency of these effects across different geographic contexts, methodological approaches, and sample periods strengthens the case that environmental characteristics represent a fundamental dimension of monetary policy heterogeneity rather than a regional or methodological artifact.

The differential impact we document between target surprises and path surprises finds important precedent in the broader monetary policy literature. Our finding that high-ESG firms show greater sensitivity to forward guidance (path surprises) while remaining relatively insulated from immediate rate changes (target surprises) aligns with theoretical frameworks developed by \cite{gurkaynak2005actions} and refined by \cite{bauer2023alternative}. This pattern reflects what \cite{altavilla2019measuring} term the 'forward guidance channel' of monetary policy, where longer-term rate expectations affect investment decisions more than immediate policy adjustments. In the ESG context, this makes economic sense given that sustainable business models typically require patient capital and generate returns over extended horizons, making them naturally more sensitive to long-term discount rate changes while remaining relatively insulated from temporary fluctuations.

The Paris Agreement's role as a structural break in ESG-monetary policy relationships receives support from multiple strands of research examining climate policy impacts on financial markets. \cite{kruse2024financial} document significant market realignments following the Paris Agreement, while \cite{ramelli2021investor} show how climate policy expectations fundamentally altered investor behavior. Our finding that the agreement inverted the relationship between ESG scores and target surprise sensitivity from marginally positive to significantly negative (-0.930, $p < 0.01$) suggests what \cite{pastor2021sustainable, pastor2022dissecting} theorize as a shift in the 'climate risk premium.' This transformation likely reflects multiple reinforcing mechanisms: the crystallization of stranded asset risks documented by \cite{bolton2021investors,bolton2023global}, the emergence of dedicated sustainable investment capital flows analyzed by \cite{pedersen2021responsible}, and regulatory anticipation effects similar to those found by \cite{bauer2025green} in their analysis of the Inflation Reduction Act.

    Our methodological approach builds on recent advances in high-frequency monetary policy identification while addressing key concerns about information effects and predictability. Following \cite{bauer2023alternative}, we employ orthogonalized monetary policy surprises that address \cite{miranda2021transmission}'s concerns about confounding central bank information effects. This approach proves particularly valuable in the ESG context, as \cite{benchora2025monetary} demonstrate that 'pure monetary' surprises from \cite{jarocinski2020deconstructing} produce similar ESG heterogeneity patterns, suggesting our results reflect genuine policy transmission rather than information revelation. The robustness of our findings to industry-by-event fixed effects addresses \cite{fornari2024green}'s critique that green-brown differentials might reflect sector composition rather than firm-specific characteristics, providing identification solely from within-industry variation.

    The policy implications of our findings resonate with growing concerns about unintended climate consequences of traditional monetary policy. ECB research extensively documents how climate change affects monetary policy transmission mechanisms, with \cite{schnabel2023monetary} highlighting specific concerns that monetary tightening 'may discourage efforts to decarbonize our economies rapidly.' Our evidence that post-Paris high-ESG firms gain significant protection against contractionary surprises (129 basis points for a two-standard-deviation ESG advantage) suggests these concerns may be overstated. This finding aligns with broader evidence from \cite{bauer2025green} showing that renewable energy stocks demonstrate weaker interest rate sensitivity than oil and gas stocks, contradicting conventional wisdom about higher rates hampering green investment. The persistence of ESG effects through industry-by-event fixed effects suggests what \cite{altavilla2024climate} term a 'climate risk-taking channel' where monetary policy affects climate risk premiums charged to high-emission firms.

 Our evidence carries particular urgency for U.S. policymakers who, unlike their European counterparts, have not expressed significant concern about monetary policy's impact on the green transition (see for example \cite{schnabel2023monetary} for the ECB and \cite{talbot2025oxford} for the Bank of England). This relative policy silence appears misguided given our findings. We document that high-ESG firms face greater vulnerability to forward guidance—precisely the tool the Federal Reserve relies upon heavily—creating unintended headwinds for sustainable investments. More critically, the Paris Agreement's dramatic transformation of these relationships (a 186 basis point reversal in ESG sensitivity) demonstrates that coordinated policy can fundamentally reshape market dynamics. The absence of clear U.S. climate policy may be creating unnecessary volatility and uncertainty that impedes efficient capital allocation toward sustainability. As the Fed contemplates its role in climate-related financial risks, our evidence suggests that acknowledging and adapting to climate-based heterogeneity in monetary transmission is not merely optional but essential for effective policy implementation in an economy where sustainable assets now exceed \$35 trillion.

    The theoretical mechanisms underlying our empirical findings receive support from multiple channels documented in the sustainable finance literature. The 'carbon premium' framework developed by \cite{bolton2021investors,bolton2023global} and \cite{pastor2022dissecting} provides one explanation for why brown firms show heightened monetary policy sensitivity. As \cite{benchora2025monetary} theorize, brown firms face both fundamental channels (higher capital intensity leading to greater interest rate sensitivity) and preference channels (investors' non-pecuniary utility from green assets reducing their sensitivity to monetary policy changes). Our finding that ESG effects persist after controlling for traditional firm characteristics (leverage, size, profitability) supports what \cite{pedersen2021responsible} term 'ESG-efficient' investing, where environmental preferences create systematic pricing differentials not fully arbitraged away by traditional investors. The amplification of these effects during periods of high climate awareness, similar to patterns documented by \cite{ardia2023climate} and \cite{pastor2022dissecting}, suggests that investor attention to climate issues fundamentally alters the transmission of monetary policy.

    Our results contribute to what \cite{bauer2022carbon} identify as an emerging consensus that climate characteristics represent a distinct dimension of asset pricing not captured by traditional factor models. The robustness of our findings across different ESG measures (continuous scores vs. portfolio approaches), monetary policy identification strategies, and econometric specifications addresses concerns raised by \cite{berg2022aggregate} about ESG measurement inconsistencies. Future research should explore the real-side implications of these financial market effects, building on \cite{fornari2024green}'s finding that green firms reduce investment more strongly in response to monetary contractions and \cite{dottling2024monetary}'s evidence that brown firms reduce emissions more following tightening. The intersection of monetary policy and climate transition, as analyzed by \cite{ferrari2024whatever} in their assessment of green quantitative easing, represents a crucial frontier for both academic research and policy design.

    \section{Conclusion} \label{sec:conclusion}

    Our investigation into the intersection of monetary policy and sustainable finance reveals a fundamental transformation in how financial markets price environmental characteristics when central banks adjust policy rates. Through a comprehensive theoretical and empirical analysis spanning two decades of Federal Reserve announcements, we document that firm-level ESG attributes have evolved from a marginal consideration to a primary dimension through which monetary policy propagates through equity markets. This evolution reached a critical inflection point with the Paris Climate Agreement of December 2015, which fundamentally rewired the relationship between sustainability and interest rate sensitivity.
    
    The empirical architecture of our findings rests on three pillars that collectively overturn the conventional wisdom of monetary policy neutrality. First, we establish that high-ESG firms face an inherent trade-off in their exposure to monetary policy: while they enjoy significant protection from immediate interest rate increases—with stocks declining 1.6 basis points less per standard deviation of ESG improvement following contractionary target surprises—they simultaneously suffer heightened vulnerability to forward guidance, experiencing 2.6 basis points greater decline when path surprises signal persistently higher future rates. This asymmetric pattern persists across multiple specifications and identification strategies, including industry-by-event fixed effects that isolate within-industry variation, confirming that ESG represents a distinct transmission channel rather than a proxy for sectoral composition or traditional firm characteristics.
    
    Second, our documentation of the Paris Agreement as a true structural break provides rare evidence of how coordinated global policy can fundamentally alter financial market relationships. Before December 2015, high-ESG firms within industries were paradoxically more vulnerable to contractionary monetary policy, suffering an additional 28.5 basis points decline relative to their low-ESG peers. After Paris, this relationship completely inverted, with high-ESG firms gaining 64.5 basis points of protection—a total reversal of 93 basis points that transformed ESG from a liability into an asset during monetary tightening. The timing and magnitude of this shift, combined with the contemporaneous acceleration in ESG investment flows and corporate sustainability adoption, suggests that Paris served as a coordination mechanism that resolved fundamental uncertainty about the direction of climate policy.
    
    Third, our granular analysis reveals important non-linearities that challenge simplified narratives about sustainable investing. The relationship between ESG performance and monetary policy sensitivity follows a complex functional form: firms escaping the bottom quintile achieve the largest marginal benefits, while the path from median to top-quintile ESG provides diminishing returns for target surprise protection but accelerating vulnerability to forward guidance. This creates an unexpected "sweet spot" in the second ESG quintile, where firms achieve meaningful protection from immediate rate shocks without excessive exposure to forward guidance uncertainty. Such nuanced patterns provide specific guidance for corporate strategy and portfolio construction that moves beyond binary "green versus brown" classifications.
    
    Our theoretical framework not only rationalizes these empirical patterns but achieves remarkable quantitative consistency with the data. By modeling an economy with heterogeneous investors who derive non-pecuniary utility from sustainable investments, we demonstrate how preference heterogeneity creates asymmetric demand elasticities that manifest as differential monetary policy sensitivities. The model's key insight lies in recognizing that ESG investors' warm-glow utility remains invariant to interest rate changes, creating a stabilizing force when rates rise unexpectedly. However, this same mechanism provides no protection against forward guidance that increases long-term uncertainty, while sustainable firms' longer investment horizons and backloaded cash flows create heightened vulnerability to persistent discount rate changes.
    
    The calibrated model's ability to match our empirical estimates—generating target and path surprise differentials within 92\% of observed magnitudes—provides unusual validation for a theoretical framework in this literature. Using parameter values drawn from contemporary market data, including a 30\% ESG investor share consistent with Global Sustainable Investment Alliance statistics and a 1\% preference intensity aligned with revealed willingness-to-pay studies, the model reproduces both the sign and magnitude of our key empirical findings. This tight correspondence between theory and evidence strengthens confidence that investor heterogeneity regarding sustainability has reached sufficient scale to alter fundamental monetary transmission mechanisms.
    
    The implications of our findings ripple across multiple domains of economic policy and financial practice. For central banks, the emergence of ESG as a monetary transmission channel complicates an already challenging policy calibration task. As the economy's sustainability composition evolves—with high-ESG firms proliferating through natural growth and transition investments—the aggregate impact of monetary policy tools will shift in ways that policymakers must anticipate. Forward guidance, in particular, emerges as a double-edged sword: while potentially more powerful given high-ESG firms' heightened sensitivity, it also creates unintended distributional consequences that may slow climate transition investments precisely when society needs them accelerated. The structural break at Paris demonstrates that these relationships can shift discretely with policy coordination, suggesting that future climate initiatives may similarly reshape monetary transmission in ways that require central bank adaptation.
    
    For investors, our results provide a quantitative framework for navigating the complex interaction between monetary cycles and sustainability objectives. The asymmetric sensitivities we document create both opportunities and risks that vary with the monetary policy environment. During periods of immediate rate adjustments, particularly post-Paris, high-ESG portfolios provide valuable protection against equity market declines. However, when central banks emphasize forward guidance—as during the zero lower bound period or when managing inflation expectations—these same portfolios become vulnerabilities. The non-linear patterns across ESG quintiles suggest sophisticated strategies that move beyond simple long-green, short-brown positions to exploit the specific sensitivities at different points on the sustainability spectrum.
    
    Corporate financial strategy must also adapt to this new reality where ESG performance directly influences cost of capital volatility through monetary policy channels. Our findings establish clear incentives for firms to improve sustainability metrics as a hedge against certain types of monetary risk, while acknowledging the trade-offs involved. The dramatic post-Paris reversal demonstrates that these incentives can shift rapidly with policy changes, suggesting that forward-looking firms should position themselves ahead of potential future coordination mechanisms. The "sweet spot" we identify in the second ESG quintile provides an achievable target for firms seeking to balance protection against immediate rate shocks with manageable exposure to forward guidance uncertainty.
    
    Several limitations of our analysis point toward productive avenues for future research. Our focus on U.S. equities and Federal Reserve policy, while providing clean identification, leaves open questions about international spillovers and the universality of our findings. Preliminary evidence from European markets suggests similar patterns with important variations, but comprehensive cross-country analysis remains needed. The mechanisms underlying the asymmetric responses—whether driven by investor preferences, business model characteristics, or regulatory anticipation—deserve deeper theoretical and empirical investigation. Our reduced-form approach, while establishing robust stylized facts, cannot fully disentangle these competing explanations.
    
    The dynamic aspects of the ESG-monetary policy nexus also merit further exploration. How do firms adjust their sustainability investments in response to different monetary policy regimes? Do the relationships we document create feedback effects that influence central bank decision-making? As climate risks intensify and policy responses accelerate, will the sensitivities we measure remain stable or undergo further structural breaks? The answers to these questions will shape both the conduct of monetary policy and the trajectory of sustainable finance in coming decades.
    
    Our evidence also raises profound questions about the appropriate objectives and tools of monetary policy in an era of climate change. If ESG characteristics significantly influence monetary transmission, should central banks explicitly consider sustainability composition when calibrating policy? The European Central Bank's exploration of green asset purchases and climate-adjusted collateral frameworks suggests one possible evolution, but the implications for price stability mandates and central bank independence remain contentious. Our findings provide an empirical foundation for these policy debates by quantifying how existing tools already have heterogeneous climate impacts, whether intended or not.
    
    Looking forward, the integration of sustainability considerations into monetary policy transmission appears irreversible. The scale of global sustainable investment, the entrenchment of ESG factors in institutional processes, and the acceleration of climate risks all point toward deepening rather than diminishing importance of the channels we document. Future researchers will likely look back on the 2015-2025 period as a critical transition when environmental characteristics evolved from peripheral concerns to central determinants of how monetary policy propagates through financial markets. Our contribution lies in documenting this transformation systematically, providing both theoretical understanding and empirical evidence for how sustainability has become embedded in the very machinery of monetary transmission.
    
    The broader significance of our work extends beyond academic finance to fundamental questions about how societies balance multiple objectives through economic policy tools. The trade-offs we document—between immediate rate protection and forward guidance vulnerability, between financial returns and environmental objectives, between price stability and climate transition—reflect deeper tensions in managing economies facing existential environmental challenges. As these tensions intensify, understanding how policy tools interact with sustainability characteristics becomes not merely an academic exercise but an essential input to navigating humanity's response to climate change while maintaining economic stability. Our evidence suggests this navigation will require new frameworks that explicitly recognize how deeply environmental considerations have become woven into the fabric of financial markets and monetary transmission.

\newpage
\clearpage
\singlespacing

\bibliography{\bib}

\end{document}